\documentclass{article}

\usepackage[a4paper,top=2.54cm,bottom=2.54cm,left=3.17cm,right=3.17cm,%
            includehead,includefoot]{geometry}

\usepackage{tikz}
\usepackage{amsmath,amssymb,amsfonts,amsthm}
\usepackage{graphicx}
\usepackage{subfigure}
\usepackage{float}
\usepackage{xcolor}
\usepackage[numbers,square,sort&compress]{natbib}
\usepackage{hyperref}
  \hypersetup{colorlinks,citecolor=blue,linkcolor=blue,breaklinks=true}
\usepackage{booktabs}
\usepackage{colortbl}
\usepackage{caption}
\usepackage{enumitem}
\usepackage{epstopdf}
\usepackage{algorithm}
\usepackage{algpseudocode}
\usepackage{array}
\usepackage{bbding}
\usepackage{fancyhdr} 
\usepackage{fancyvrb} 
\usepackage{longtable}
\usepackage{listings}
\usepackage{rotating,rotfloat} 
\usepackage{yhmath} 
\usepackage{quantikz}
\usepackage{braket}

\usepackage{fancyhdr}
\allowdisplaybreaks

\newtheorem{theorem}{Theorem}[section]
\newtheorem{lemma}{Lemma}[section]
\newtheorem{corollary}[theorem]{Corollary}
\newtheorem{proposition}[theorem]{Proposition}
\newtheorem{definition}{Definition}[section]
\newtheorem{example}{Example}
\newtheorem{remark}{Remark}

\newcommand{\bbm}{\begin{bmatrix}}
\newcommand{\ebm}{\end{bmatrix}}

\begin{document}

\title{Protected Logical Qudits in Kitaev Quantum Double Models via Stable Representations}

\author{Naihong Hu$^*$, Futao Wang}

\footnotetext{Corresponding author. This work is supported by the Fundamental and Interdisciplinary Disciplines Breakthrough Plan of the Ministry of Education of China (JYB2025XDXM112), the National Natural Science Foundation of China (Grant No.~12171155), and in part by the Science and Technology Commission of Shanghai Municipality (Grant No.~22DZ2229014).}

\maketitle

\begin{abstract}
Fault-tolerant quantum computation requires robust protection of encoded quantum information. In this work, we develop a representation-theoretic framework for constructing protected logical qudits in finite-group Kitaev quantum double models. By introducing $\varepsilon$-stable irreducible representations, we establish a necessary and sufficient existence criterion and derive ribbon--projector commutation relations that yield a $d$-dimensional protected logical subspace. We apply this construction to symmetric and alternating groups, obtaining logical qubits for $S_n$ ($n\ge 3$) and a logical qutrit for $A_4$. Moreover, the family $(\mathbb Z_2)^d\rtimes\mathbb Z_d$ realizes protected logical qudits of arbitrary dimension $d\ge 2$. Finally, for $D(A_4)$, we describe a scheme for universal logical qutrit computation using ribbon-based logical operations.
\end{abstract}
\section{Introduction}\label{sec:intro}

Reliable quantum computation requires protecting encoded quantum information against decoherence, control imperfections, and other sources of noise. Quantum error correction provides the mathematical framework for this task. In particular, the Knill--Laflamme conditions \cite{KL97} give necessary and sufficient conditions for a subspace to serve as a quantum error-correcting code, while standard constructions and general principles of quantum error correction are reviewed systematically in \cite{NL12}. Topological approaches to quantum error correction are especially attractive because logical information is encoded non-locally and can therefore be insensitive to suitable classes of local perturbations.

Kitaev's quantum double construction provides one of the fundamental frameworks for topological quantum error correction. For a finite group $G$, the corresponding quantum double model is an exactly solvable lattice model whose quasiparticle excitations are governed by the representation theory of the Drinfeld double $D(G)$ \cite{Kit03}. In the Abelian case, one recovers toric-code-type models, while non-Abelian groups give rise to a richer excitation structure and nontrivial fusion and braiding phenomena.

The quantum double framework has subsequently been generalized in several algebraic and physical directions. Boundary structures and condensation phenomena have been studied in quantum double models \cite{BK98,BSW11}. The finite-group construction has also been extended to semisimple Hopf algebras \cite{BMCA13}, groupoid-based formulations \cite{Cha13}, and weak Hopf algebra settings \cite{JTKC23}. In parallel, representation-theoretic and categorical methods have clarified the relationship between Hopf algebra structures, quantum states, and diagrammatic quantum computation \cite{CM22,CD11,Maj22}. Quantum double models are also closely related to lattice gauge theory \cite{Meu17}. These developments motivate a systematic study of how finite-group representation theory can be used not only to classify quasiparticle excitations, but also to construct logical quantum degrees of freedom and logical operations.

Non-Abelian quantum double models are particularly natural candidates for quantum information processing because of their richer anyonic excitation structures. Previous work has investigated quantum error correction in the $D(S_3)$ model \cite{WLP09}. More recently, universal logical gates based on braiding and fusion of non-Abelian anyons in a $D(S_3)$-type setting have been demonstrated on quantum hardware \cite{LLG26}. These developments motivate the search for general group-theoretic conditions under which finite-group quantum double models support protected logical qudits.

In this work, we develop a general representation-theoretic framework for constructing protected logical qudits in finite-group Kitaev quantum double models. For a nontrivial one-dimensional irreducible representation $\varepsilon$ of order $d$, we introduce the notion of an $\varepsilon$-stable irreducible representation $\rho$, satisfying
$
\rho\otimes\varepsilon\simeq\rho,
$
and establish a necessary and sufficient criterion for its existence. We then show that such a representation gives rise to a $d$-dimensional protected logical subspace, with the associated ribbon trace operator implementing the logical cyclic shift. Errors that change the local quasiparticle sector are detectable by quasiparticle-sector measurements.

We apply this construction to several families of finite groups. For $S_n$ with $n\ge 3$, it yields logical qubits, while among the alternating groups, the proposed mechanism singles out $A_4$, which yields a logical qutrit. More generally, the family $(\mathbb Z_2)^d\rtimes\mathbb Z_d$ realizes logical qudits of arbitrary dimension $d\ge 2$. Finally, we specialize to $D(A_4)$ and describe a scheme for universal logical qutrit computation.

The remainder of this paper is organized as follows. Section~\ref{sec:prereq} reviews the algebraic and lattice-theoretic framework of the Kitaev quantum double model, including quasiparticle projectors and ribbon operators. Section~\ref{sec:construction} develops the general logical qudit construction and its representation-theoretic existence criterion, applies the theory to symmetric and alternating groups and to an explicit family that realizes arbitrary logical dimensions, and finally specializes to universal logical qutrit computation in the $D(A_4)$ model.

\section{Preliminaries}\label{sec:prereq}

This section briefly reviews the theoretical framework of the Kitaev quantum double model associated with a finite group~$G$.

\subsection{Group Hopf algebra and its quantum double structure}
\medskip
\noindent
\textbf{(I) The group algebra $\mathbb{C}G$.}
The group algebra $\mathbb{C}G$ is the complex vector space with basis $\{h\}_{h\in G}$. Its Hopf algebra structure is defined on basis elements by $m(h\otimes k)=hk$, $\eta(\lambda)=\lambda e$, $\Delta(h)=h\otimes h$, $\varepsilon(h)=1$, and $S(h)=h^{-1}$, and extended linearly.

\medskip
\noindent
\textbf{(II) The dual Hopf algebra $\mathbb{C}(G)$.}
The dual Hopf algebra $\mathbb{C}(G)$ is the complex vector space spanned by the characteristic functions $\{\delta_g\}_{g\in G}$. Its multiplication is pointwise, $\delta_g\delta_h=\delta_{g,h}\delta_g$, with unit $\eta(\lambda)=\lambda\sum_{g\in G}\delta_g$. The comultiplication, counit, and antipode are given by $\Delta(\delta_g)=\sum_{h\in G}\delta_h\otimes\delta_{h^{-1}g}$, $\varepsilon(\delta_g)=\delta_{g,e}$, and $S(\delta_g)=\delta_{g^{-1}}$, respectively.

\medskip
\noindent
\textbf{(III) The quantum double $D(G)$.}
The quantum double of $G$, denoted by $D(G):=\mathbb{C}(G)\rtimes\mathbb{C}G$ \cite{Maj95}, is a quasitriangular Hopf algebra obtained as the smash product of $\mathbb{C}(G)$ and $\mathbb{C}G$. As a vector space, $D(G)\cong\mathbb{C}(G)\otimes\mathbb{C}G$. Its multiplication is determined by the crossed relation $(1\otimes h)(\delta_g\otimes1)=\delta_{hgh^{-1}}\otimes h$, which yields
$
(\delta_g\otimes h)(\delta_{g'}\otimes h')
=
\delta_{g,hg'h^{-1}}\delta_g\otimes hh'.
$
The comultiplication and antipode are given by $\Delta(\delta_g\otimes h)=\sum_{k\in G}(\delta_k\otimes h)\otimes(\delta_{k^{-1}g}\otimes h)$ and $S(\delta_g\otimes h)=\delta_{h^{-1}g^{-1}h}\otimes h^{-1}$, respectively.

The normalized integrals of $\mathbb{C}G$ and $\mathbb{C}(G)$ are $\Lambda=\frac{1}{|G|}\sum_{h\in G}h$ and $\Lambda^*=\delta_e$, respectively. The universal $\mathcal{R}$-matrix is $\mathcal{R}=\sum_{h\in G}(\delta_h\otimes1)\otimes(1\otimes h)$, which induces a braiding $\Psi_{V,W}:V\otimes W\to W\otimes V$ on $D(G)$-modules $V$ and $W$ via $\Psi_{V,W}(v\otimes w)=\tau\!\left(\mathcal{R}\cdot(v\otimes w)\right)$, where $\tau(x\otimes y)=y\otimes x$ denotes the flip map.

 \subsection{Kitaev quantum double model based on group \texorpdfstring{$G$}{G}}

Kitaev's quantum double model is defined on a lattice $\mathcal{L}$ embedded in a closed oriented surface $\Sigma$. The lattice $\mathcal{L}$ consists of vertices, edges, and faces (or plaquettes). Each edge $e$ is assigned an arbitrary but fixed orientation. Every edge connects two vertices, denoted by $v_1$ and $v_2$, which are said to be incident to $e$ and adjacent to each other. Moreover, each edge belongs to the boundary of exactly two faces. A representative lattice configuration is illustrated below.
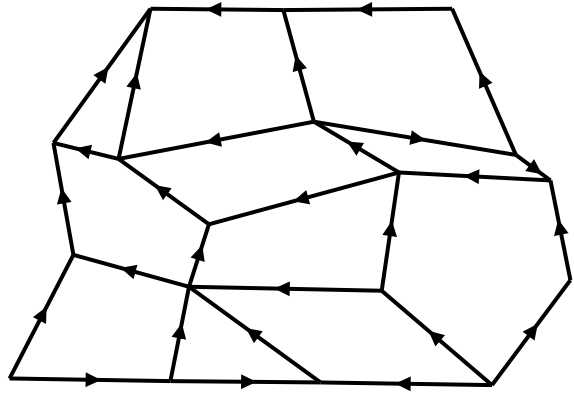
\[
\begin{tikzpicture}[x=0.5pt,y=0.5pt,yscale=-1,xscale=1]

\draw [line width=1.5]    (205.2,49.9) -- (305.2,50.9) ;
\draw [shift={(246.9,50.32)}, rotate = 0.57] [fill={rgb, 255:red, 0; green, 0; blue, 0 }  ][line width=0.08]  [draw opacity=0] (11.61,-5.58) -- (0,0) -- (11.61,5.58) -- cycle    ;
\draw [line width=1.5]    (305.2,50.9) -- (432.2,49.9) ;
\draw [shift={(360.4,50.47)}, rotate = 359.55] [fill={rgb, 255:red, 0; green, 0; blue, 0 }  ][line width=0.08]  [draw opacity=0] (11.61,-5.58) -- (0,0) -- (11.61,5.58) -- cycle    ;
\draw [line width=1.5]    (205.2,49.9) -- (132.2,150.9) ;
\draw [shift={(173.56,93.67)}, rotate = 125.86] [fill={rgb, 255:red, 0; green, 0; blue, 0 }  ][line width=0.08]  [draw opacity=0] (11.61,-5.58) -- (0,0) -- (11.61,5.58) -- cycle    ;
\draw [line width=1.5]    (205.2,49.9) -- (181.2,162.9) ;
\draw [shift={(194.92,98.28)}, rotate = 101.99] [fill={rgb, 255:red, 0; green, 0; blue, 0 }  ][line width=0.08]  [draw opacity=0] (11.61,-5.58) -- (0,0) -- (11.61,5.58) -- cycle    ;
\draw [line width=1.5]    (181.2,162.9) -- (328.2,134.9) ;
\draw [shift={(246.55,150.45)}, rotate = 349.22] [fill={rgb, 255:red, 0; green, 0; blue, 0 }  ][line width=0.08]  [draw opacity=0] (11.61,-5.58) -- (0,0) -- (11.61,5.58) -- cycle    ;
\draw [line width=1.5]    (305.2,50.9) -- (328.2,134.9) ;
\draw [shift={(314.51,84.89)}, rotate = 74.69] [fill={rgb, 255:red, 0; green, 0; blue, 0 }  ][line width=0.08]  [draw opacity=0] (11.61,-5.58) -- (0,0) -- (11.61,5.58) -- cycle    ;
\draw [line width=1.5]    (432.2,49.9) -- (480.2,159.9) ;
\draw [shift={(452.88,97.29)}, rotate = 66.43] [fill={rgb, 255:red, 0; green, 0; blue, 0 }  ][line width=0.08]  [draw opacity=0] (11.61,-5.58) -- (0,0) -- (11.61,5.58) -- cycle    ;
\draw [line width=1.5]    (480.2,159.9) -- (328.2,134.9) ;
\draw [shift={(412.39,148.75)}, rotate = 189.34] [fill={rgb, 255:red, 0; green, 0; blue, 0 }  ][line width=0.08]  [draw opacity=0] (11.61,-5.58) -- (0,0) -- (11.61,5.58) -- cycle    ;
\draw [line width=1.5]    (132.2,150.9) -- (181.2,162.9) ;
\draw [shift={(148.64,154.93)}, rotate = 13.76] [fill={rgb, 255:red, 0; green, 0; blue, 0 }  ][line width=0.08]  [draw opacity=0] (11.61,-5.58) -- (0,0) -- (11.61,5.58) -- cycle    ;
\draw [line width=1.5]    (181.2,162.9) -- (249.2,212.1) ;
\draw [shift={(208.48,182.63)}, rotate = 35.89] [fill={rgb, 255:red, 0; green, 0; blue, 0 }  ][line width=0.08]  [draw opacity=0] (11.61,-5.58) -- (0,0) -- (11.61,5.58) -- cycle    ;
\draw [line width=1.5]    (249.2,212.1) -- (392.2,173.1) ;
\draw [shift={(312.69,194.78)}, rotate = 344.74] [fill={rgb, 255:red, 0; green, 0; blue, 0 }  ][line width=0.08]  [draw opacity=0] (11.61,-5.58) -- (0,0) -- (11.61,5.58) -- cycle    ;
\draw [line width=1.5]    (328.2,134.9) -- (392.2,173.1) ;
\draw [shift={(353.07,149.75)}, rotate = 30.83] [fill={rgb, 255:red, 0; green, 0; blue, 0 }  ][line width=0.08]  [draw opacity=0] (11.61,-5.58) -- (0,0) -- (11.61,5.58) -- cycle    ;
\draw [line width=1.5]    (392.2,173.1) -- (506.2,179.1) ;
\draw [shift={(440.91,175.66)}, rotate = 3.01] [fill={rgb, 255:red, 0; green, 0; blue, 0 }  ][line width=0.08]  [draw opacity=0] (11.61,-5.58) -- (0,0) -- (11.61,5.58) -- cycle    ;
\draw [line width=1.5]    (506.2,179.1) -- (480.2,159.9) ;
\draw [shift={(499.88,174.43)}, rotate = 216.44] [fill={rgb, 255:red, 0; green, 0; blue, 0 }  ][line width=0.08]  [draw opacity=0] (11.61,-5.58) -- (0,0) -- (11.61,5.58) -- cycle    ;
\draw [line width=1.5]    (506.2,179.1) -- (521.2,254.1) ;
\draw [shift={(512.07,208.46)}, rotate = 78.69] [fill={rgb, 255:red, 0; green, 0; blue, 0 }  ][line width=0.08]  [draw opacity=0] (11.61,-5.58) -- (0,0) -- (11.61,5.58) -- cycle    ;
\draw [line width=1.5]    (521.2,254.1) -- (462.2,333.1) ;
\draw [shift={(496.67,286.95)}, rotate = 126.75] [fill={rgb, 255:red, 0; green, 0; blue, 0 }  ][line width=0.08]  [draw opacity=0] (11.61,-5.58) -- (0,0) -- (11.61,5.58) -- cycle    ;
\draw [line width=1.5]    (392.2,173.1) -- (379.2,262.1) ;
\draw [shift={(386.9,209.39)}, rotate = 98.31] [fill={rgb, 255:red, 0; green, 0; blue, 0 }  ][line width=0.08]  [draw opacity=0] (11.61,-5.58) -- (0,0) -- (11.61,5.58) -- cycle    ;
\draw [line width=1.5]    (379.2,262.1) -- (462.2,333.1) ;
\draw [shift={(414.39,292.2)}, rotate = 40.54] [fill={rgb, 255:red, 0; green, 0; blue, 0 }  ][line width=0.08]  [draw opacity=0] (11.61,-5.58) -- (0,0) -- (11.61,5.58) -- cycle    ;
\draw [line width=1.5]    (249.2,212.1) -- (234.2,259.1) ;
\draw [shift={(244.22,227.69)}, rotate = 107.7] [fill={rgb, 255:red, 0; green, 0; blue, 0 }  ][line width=0.08]  [draw opacity=0] (11.61,-5.58) -- (0,0) -- (11.61,5.58) -- cycle    ;
\draw [line width=1.5]    (234.2,259.1) -- (379.2,262.1) ;
\draw [shift={(298.4,260.43)}, rotate = 1.19] [fill={rgb, 255:red, 0; green, 0; blue, 0 }  ][line width=0.08]  [draw opacity=0] (11.61,-5.58) -- (0,0) -- (11.61,5.58) -- cycle    ;
\draw [line width=1.5]    (234.2,259.1) -- (333.2,331.1) ;
\draw [shift={(276.99,290.22)}, rotate = 36.03] [fill={rgb, 255:red, 0; green, 0; blue, 0 }  ][line width=0.08]  [draw opacity=0] (11.61,-5.58) -- (0,0) -- (11.61,5.58) -- cycle    ;
\draw [line width=1.5]    (333.2,331.1) -- (462.2,333.1) ;
\draw [shift={(389.4,331.97)}, rotate = 0.89] [fill={rgb, 255:red, 0; green, 0; blue, 0 }  ][line width=0.08]  [draw opacity=0] (11.61,-5.58) -- (0,0) -- (11.61,5.58) -- cycle    ;
\draw [line width=1.5]    (132.2,150.9) -- (147.2,235.1) ;
\draw [shift={(138.24,184.83)}, rotate = 79.9] [fill={rgb, 255:red, 0; green, 0; blue, 0 }  ][line width=0.08]  [draw opacity=0] (11.61,-5.58) -- (0,0) -- (11.61,5.58) -- cycle    ;
\draw [line width=1.5]    (147.2,235.1) -- (234.2,259.1) ;
\draw [shift={(182.7,244.89)}, rotate = 15.42] [fill={rgb, 255:red, 0; green, 0; blue, 0 }  ][line width=0.08]  [draw opacity=0] (11.61,-5.58) -- (0,0) -- (11.61,5.58) -- cycle    ;
\draw [line width=1.5]    (234.2,259.1) -- (220.2,330.1) ;
\draw [shift={(228.81,286.46)}, rotate = 101.15] [fill={rgb, 255:red, 0; green, 0; blue, 0 }  ][line width=0.08]  [draw opacity=0] (11.61,-5.58) -- (0,0) -- (11.61,5.58) -- cycle    ;
\draw [line width=1.5]    (333.2,331.1) -- (220.2,330.1) ;
\draw [shift={(285,330.67)}, rotate = 180.51] [fill={rgb, 255:red, 0; green, 0; blue, 0 }  ][line width=0.08]  [draw opacity=0] (11.61,-5.58) -- (0,0) -- (11.61,5.58) -- cycle    ;
\draw [line width=1.5]    (147.2,235.1) -- (99.2,328.1) ;
\draw [shift={(127.01,274.22)}, rotate = 117.3] [fill={rgb, 255:red, 0; green, 0; blue, 0 }  ][line width=0.08]  [draw opacity=0] (11.61,-5.58) -- (0,0) -- (11.61,5.58) -- cycle    ;
\draw [line width=1.5]    (220.2,330.1) -- (99.2,328.1) ;
\draw [shift={(168,329.24)}, rotate = 180.95] [fill={rgb, 255:red, 0; green, 0; blue, 0 }  ][line width=0.08]  [draw opacity=0] (11.61,-5.58) -- (0,0) -- (11.61,5.58) -- cycle    ;

\end{tikzpicture}
\]

For each edge $e \in E$, we can define a Hilbert space $\mathcal{H}_e = \mathbb{C}[G] = \operatorname{span}\{|g\rangle : g \in G\}$, meaning that each edge carries a qudit whose computational basis is labelled by the elements of \(G\). The total Hilbert space for the lattice $\mathcal{L}$ is then defined as $\mathcal{H}_{\text{tot}} = \bigotimes_{e \in E} \mathcal{H}_e$.

A site $s$ is defined as an ordered pair $(v, p) \in V \times P$ where $v \subset p$ (i.e. vertex $v$ lies on the boundary of plaquette $p$). For each site $s$ and group element $g \in G$, we define two types of local operators: $A_g(s)$ and $B_h(s)$.

\begin{definition}[Local Operators]

A site $s=(v,p)\in V\times P$ with $v\in\partial p$ admits:

(I)  Operators $A_g(s) \in End(\mathcal{H}_{\mathrm{tot}})$: Acts on $\mathrm{star}(v)$ edges through orientation-dependent conjugation:
  \[
A_g(s) = \bigotimes_{e\in\mathrm{in}(v)} R_g^{(e)} \otimes \bigotimes_{e\in\mathrm{out}(v)} L_g^{(e)}
\]
  where $L_g|h\rangle = |gh\rangle$, $R_g|h\rangle = |hg^{-1}\rangle$
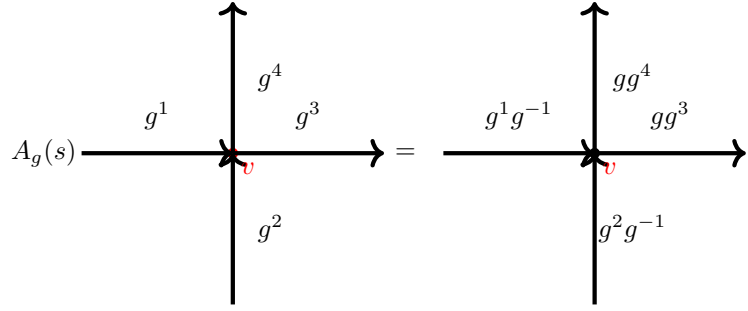
\[
\begin{tikzpicture}
 \fill [red](0,0) circle (2pt);
\draw [->,ultra thick](0,0) -- (2,0);
\draw [->,ultra thick](0,0) -- (0,2);
\draw [->,ultra thick](-2,0) -- (0,0);
\draw [->,ultra thick](0,-2) -- (0,0);
\node [red]at (0.2, -0.2) {$v$};
\node at (-1, 0.5) {$g^1$};
\node at (1, 0.5) {$g^3$};
\node at (0.5, 1) {$g^4$};
\node at (0.5, -1) {$g^2$};
\node at (-2.5, 0) {$A_g(s)$};
\end{tikzpicture}
\begin{tikzpicture}[xshift=4]
 \fill (0,0) circle (2pt);
\draw [->,ultra thick](0,0) -- (2,0);
\draw [->,ultra thick](0,0) -- (0,2);
\draw [->,ultra thick](-2,0) -- (0,0);
\draw [->,ultra thick](0,-2) -- (0,0);
\node [red]at (0.2, -0.2) {$v$};
\node at (-1, 0.5) {$g^1g^{-1}$};
\node at (1, 0.5) {$gg^3$};
\node at (0.5, 1) {$gg^4$};
\node at (0.5, -1) {$g^2g^{-1}$};
\node at (-2.5, 0) {$=$};
\end{tikzpicture}
\]
(II)  Operators $B_h(s) \in End(\mathcal{H}_{\mathrm{tot}})$:
 Starting from the vertex $v$,  multiply the group elements  along $\partial p$ in a clockwise direction:
  \[
B_h(s) = \delta_{\left(\prod_{e\in\partial p} g_e^{\epsilon(e)}, h\right)}
\]
  with $\epsilon(e)=+1$ ($-1$) if edge orientation matches (opposes) traversal

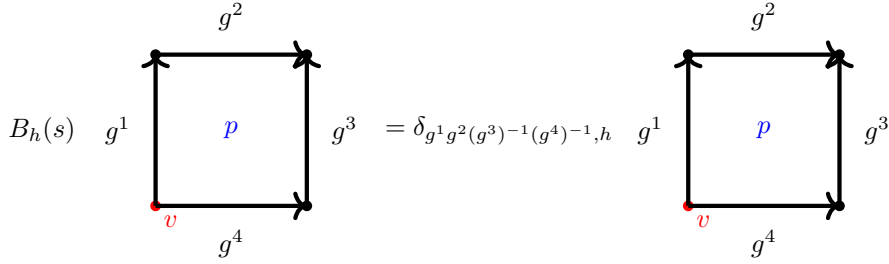
\begin{tikzpicture}
 \fill[red] (0,0) circle (2pt);
 \fill (0,2) circle (2pt);
 \fill (2,2) circle (2pt);
 \fill (2,0) circle (2pt);
\draw [->,ultra thick](0,0) -- (2,0);
\draw [->,ultra thick](0,0) -- (0,2);
\draw [->,ultra thick](2,0) -- (2,2);
\draw [->,ultra thick](0,2) -- (2,2);
\node [red]at (0.2, -0.2) {$v$};
\node at (-0.5, 1) {$g^1$};
\node at (2.5,1) {$g^3$};
\node at (1, -0.5) {$g^4$};
\node at (1, 2.5) {$g^2$};
\node at (-1.5, 1) {$B_h(s)$};
\node [blue] at (1, 1) {$p$};
\end{tikzpicture}
\begin{tikzpicture}[xshift=4]
 \fill[red] (0,0) circle (2pt);
 \fill (0,2) circle (2pt);
 \fill (2,2) circle (2pt);
 \fill (2,0) circle (2pt);
\draw [->,ultra thick](0,0) -- (2,0);
\draw [->,ultra thick](0,0) -- (0,2);
\draw [->,ultra thick](2,0) -- (2,2);
\draw [->,ultra thick](0,2) -- (2,2);
\node[red] at (0.2, -0.2) {$v$};
\node at (-0.5, 1) {$g^1$};
\node at (2.5,1) {$g^3$};
\node at (1, -0.5) {$g^4$};
\node at (1, 2.5) {$g^2$};
\node at (-2.5, 1) {$=\delta_{g^1g^2(g^3)^{-1}(g^4)^{-1},h}$};
\node [blue] at (1, 1) {$p$};
\end{tikzpicture}

\end{definition}
\begin{proposition}

The operators $A_g(s)$ and $B_h(s)$ satisfy the following relations:

1.
   $
   A_{g_1}A_{g_2} = A_{g_1g_2}, \quad A_e = \mathbb{I}
   $

2.
   $
   B_{h_1}B_{h_2} = \delta_{h_1,h_2}B_{h_1}
   $

3.
   $
   A_gB_h = B_{ghg^{-1}}A_g
  $

4.
$A_e(s)=\sum \limits_{h\in G}  B_h(s)=Id$
	
	$\forall {g_1},{g_2,h_1},{h_2}\in G$.
	
\end{proposition}	
 The ground state subspace $\mathcal{H}_{\text{vac}} \subset \mathcal{H}_{\text{tot}}$ is:
\[
\mathcal{H}_{vac}=\{ |\psi \rangle \in \mathcal H_{tot} : A(v)|\psi \rangle=|\psi \rangle,B(p)|\psi \rangle=|\psi \rangle,\forall v\in V,p\in P \}
\]
where
\[
A(v) = \frac{1}{|G|}\sum_{g\in G}A_g(v), \quad B(p) = B_e(p)
\]

\begin{theorem} [D(G)-Module Structure]

The local operator actions:
\[
\triangleright: D(G) \to \mathrm{End}(\mathcal{H}_{\mathrm{tot}}),\quad (\delta_h\otimes g) \triangleright \mapsto B_h(s)A_g(s)
\]
define a  representation of Drinfeld double $D(G)$
\end{theorem}
\begin{proof}
There is a direct calculation acting on the six relevant edge Hilbert spaces.
\[
\begin{tikzpicture}
 \fill [red](0,0) circle (2pt);
\draw [->,ultra thick](0,0) -- (2,0);
\draw [->,ultra thick](0,0) -- (0,2);
\draw [->,ultra thick](-2,0) -- (0,0);
\draw [->,ultra thick](0,-2) -- (0,0);

\node at (-1, 0.5) {$g^1$};
\node at (1, 0.5) {$g^3$};
\node at (0.5, 1) {$g^4$};
\node at (0.5, -1) {$g^2$};

 \fill[red] (0,0) circle (2pt);
 \fill (0,2) circle (2pt);
 \fill (2,2) circle (2pt);
 \fill (2,0) circle (2pt);

\draw [->,ultra thick](2,0) -- (2,2);
\draw [->,ultra thick](0,2) -- (2,2);
\node [red]at (0.2, -0.2) {$v$};
\node at (2.5,1) {$g^6$};
\node at (1, 2.5) {$g^5$};
\node [blue] at (1, 1) {$p$};

\end{tikzpicture}
\begin{tikzpicture}
 \fill [red](0,0) circle (2pt);
\draw [->,ultra thick](0,0) -- (2,0);
\draw [->,ultra thick](0,0) -- (0,2);
\draw [->,ultra thick](-2,0) -- (0,0);
\draw [->,ultra thick](0,-2) -- (0,0);

\node at (-1, 0.5) {$g^1$};
\node at (1, 0.5) {$g^3$};
\node at (0.5, 1) {$g^4$};
\node at (0.5, -1) {$g^2$};

 \fill[red] (0,0) circle (2pt);
 \fill (0,2) circle (2pt);
 \fill (2,2) circle (2pt);
 \fill (2,0) circle (2pt);
\draw [->,ultra thick](-5.5,0) -- (-2.5,0);
\draw [->,ultra thick](2,0) -- (2,2);
\draw [->,ultra thick](0,2) -- (2,2);
\node [red]at (0.2, -0.2) {$v$};
\node at (2.5,1) {$g^6$};
\node at (1, 2.5) {$g^5$};
\node [blue] at (1, 1) {$p$};
\node  at (-4, 0.5) {$\delta_{h^{-1}gh}\triangleright$};
\node at (-2.5, -1) {$\delta_{h^{-1}gh}(g^4g^5(g^6)^{-1}(g^3)^{-1})$};
\end{tikzpicture}
\]
\[
\begin{tikzpicture}
 \fill [red](0,0) circle (2pt);
\draw [->,ultra thick](0,0) -- (2,0);
\draw [->,ultra thick](0,0) -- (0,2);
\draw [->,ultra thick](-2,0) -- (0,0);
\draw [->,ultra thick](0,-2) -- (0,0);

\node at (-1, 0.5) {$g^1h^{-1}$};
\node at (1, 0.5) {$hg^3$};
\node at (0.5, 1) {$hg^4$};
\node at (0.5, -1) {$g^2h^{-1}$};

 \fill[red] (0,0) circle (2pt);
 \fill (0,2) circle (2pt);
 \fill (2,2) circle (2pt);
 \fill (2,0) circle (2pt);
\draw [->,ultra thick](0,4.5) -- (0,2.5);
\draw [->,ultra thick](2,0) -- (2,2);
\draw [->,ultra thick](0,2) -- (2,2);
\node [red]at (0.2, -0.2) {$v$};
\node at (2.5,1) {$g^6$};
\node at (1, 2.5) {$g^5$};
\node [blue] at (1, 1) {$p$};
\node  at (0.5, 3.5) {$h\triangleright$};
\end{tikzpicture}
\begin{tikzpicture}
 \fill [red](0,0) circle (2pt);
\draw [->,ultra thick](0,0) -- (2,0);
\draw [->,ultra thick](0,0) -- (0,2);
\draw [->,ultra thick](-2,0) -- (0,0);
\draw [->,ultra thick](0,-2) -- (0,0);

\node at (-1, 0.5) {$g^1h^{-1}$};
\node at (1, 0.5) {$hg^3$};
\node at (0.5, 1) {$hg^4$};
\node at (0.5, -1) {$g^2h^{-1}$};

 \fill[red] (0,0) circle (2pt);
 \fill (0,2) circle (2pt);
 \fill (2,2) circle (2pt);
 \fill (2,0) circle (2pt);
\draw [->,ultra thick](-5.5,0) -- (-2.5,0);
\draw [->,ultra thick](2,0) -- (2,2);
\draw [->,ultra thick](0,2) -- (2,2);
\draw [->,ultra thick](0,4.5) -- (0,2.5);
\node [red]at (0.2, -0.2) {$v$};
\node at (2.5,1) {$g^6$};
\node at (1, 2.5) {$g^5$};
\node [blue] at (1, 1) {$p$};
\node  at (-4, 0.5) {$\delta_{g}\triangleright$};
\node at (-2.5, -1) {$\delta_{g}(hg^4g^5(g^6)^{-1}(g^3)^{-1}h^{-1})$};
\node  at (0.5, 3.5) {$h\triangleright$};
\end{tikzpicture}
\]
\end{proof}

\begin{theorem}\cite{CDHP20}{ (Topological Ground State Degeneracy)}

Let $\Sigma$ be a closed oriented surface of genus $g$. The degeneracy of the quantum double model's ground state space is topologically protected and satisfies:
\[
\dim \mathcal{H}_{\mathrm{vac}} = \left| \mathrm{Hom}(\pi_1(\Sigma), G)/G \right|
\]
where
 $\mathrm{Hom}(\pi_1(\Sigma), G)$ denotes group homomorphisms from the fundamental group to $G$ .
 The quotient $/G$ identifies orbits under the adjoint action $h \triangleright \rho = h\rho h^{-1}$ for $h \in G$, where $\rho \in \mathrm{Hom}(\pi_1(\Sigma), G)$
\end{theorem}

\begin{corollary}\cite{CDHP20} (Planar Uniqueness)

When $\Sigma$ is a contractible plane ($g=0$), the ground state is unique:
\[
| \mathrm{vac} \rangle = \prod_{v \in V} A(v) \left( \bigotimes_{e \in E} |e\rangle \right)
\]
where $e$ is the identity element of the group $G$, and the tensor product $\bigotimes$ runs over all edges.

\end{corollary}

\begin{example}(Toric Code Case)

For $G = \mathbb{Z}_2$ on torus ($g=1$):
\[
\dim \mathcal{H}_{\mathrm{vac}} = \left| \mathrm{Hom}(\mathbb{Z}^2, \mathbb{Z}_2)/\mathbb{Z}_2 \right| = 4
\]
Matching the well-known $4$-fold degeneracy of toric code.
\end{example}
\subsection{Quasiparticle and projection operator}
In complex many-body interacting systems, quasiparticles emerge as effective low-energy excitations. Within the framework of the Kitaev quantum double model, these quasiparticles are mathematically characterized as irreducible representations of the Drinfeld double $D(G)$. When the system resides in an excited state, its excitations can be systematically described through these representations.

The irreducible representations of $D(G)$ are parameterized by pairs $(\mathcal{C}, \pi)$, where $\mathcal{C}$ denotes a conjugacy class of $G$ (interpreted as magnetic flux), and $\pi$ is an irreducible representation of the centralizer subgroup $C_G$ of a chosen representative $r_{\mathcal{C}} \in \mathcal{C}$ (interpreted as electric charge). The choice of $r_\mathcal{C}$ affects the centralizer subgroup $C_G$, but $C_G$ remains unique up to isomorphism.

In the case of the generalized toric code (when the group $G$ is $\mathbb{Z}_n$), the corresponding structures are as follows:

\noindent \textbf{Fluxons} $(\mathcal{C}, 1)$:
Acting on $\mathbb{C}\mathcal{C}$,
\[
\delta_g h \triangleright c = \delta_{g,\, h c h^{-1}} h c h^{-1}.
\]

\noindent \textbf{Chargeons} $(\{e\}, \pi)$:
Acting on the representation space $V_\pi$,
\[
\delta_g h \triangleright w = \delta_{g,e} \pi(h) w.
\]

The nontrivial anyonic braiding between excitations is encoded in the $\mathcal{R}$-matrix of $D(G)$:
\begin{enumerate}
    \item \emph{Flux-flux braiding:}
    \[
\Psi(f \otimes f') = \sum_{g \in G} (g \triangleright f') \otimes (\delta_g \triangleright f) = f f' f^{-1} \otimes f.
\]

    \item \emph{Flux-charge braiding:}
    \[
\Psi(f \otimes w) = \sum_{g \in G} (g \triangleright w) \otimes (\delta_g \triangleright f) = \pi(f) w \otimes f.
\]
\end{enumerate}

\begin{remark}
For Abelian groups ($G = \mathbb{Z}_n$), the braiding simplifies to phase factors, recovering the generalized toric code anyons.
\end{remark}
The irreducible representations $(\mathcal{C}, \pi)$ of $D(G)$ can be described as follows \cite{Maj04}:

\begin{enumerate}
    \item Choose a map $q: \mathcal{C} \to G$ satisfying:
    \[
q_c r_{\mathcal{C}} q_c^{-1} = c, \quad \forall c \in \mathcal{C},
\]
    where $r_{\mathcal{C}} \in \mathcal{C}$ is a fixed representative.

    \item Define a cocycle $\zeta: \mathcal{C} \times G \to C_G$ via:
    \[
\zeta_c(g) := q_{g c g^{-1}}^{-1} g q_c \in C_G(r_{\mathcal{C}}).
\]
    This satisfies the cocycle condition:
    \[
\zeta_c(gh) = \zeta_{h c h^{-1}}(g) \cdot \zeta_c(h), \quad \forall g,h \in G.
\]
\end{enumerate}

The $D(G)$-action on $\mathbb{C}\mathcal{C} \otimes V_\pi$ is given by:
\[
\delta_g h \triangleright (c \otimes w) = \delta_{g,\, h c h^{-1}} h c h^{-1} \otimes \pi\left(\zeta_c(h)\right) w.
\]
This defines an irreducible representation. Different choices of $q_c$ yield isomorphic representations. By choosing $q_{r_{\mathcal{C}}} = e$, we recover the chargeon representation rather than an equivalent conjugate of it.

To describe the projection operators for detecting the existence of quasiparticles, we first concentrate on the chargeon sector. For each irreducible representation $\pi$, such quasiparticles can be detected by measuring the observable
\[
O = \sum_{\pi} r_{\pi} P_{\pi}(s) = \sum_{\pi} r_{\pi} P_{\pi} \triangleright_{s},
\quad \text{where } \forall \pi, \, r_{\pi} \in \mathbb{R},
\]
$s$ is a site, and $P_{\pi}$ is the central projection element (central idempotent) in the group algebra $\mathbb{C}G$, given by
\[
P_{\pi} = \frac{\dim V_{\pi}}{|G|} \sum_{g\in G} \bigl(\operatorname{Tr}\pi(g^{-1})\bigr) g.
\]
The centrality, idempotency, and orthogonality follow from the orthogonality relations of finite group characters:
\[
\begin{aligned}
\sum_{h\in G} \operatorname{Tr}\pi(h^{-1}) \operatorname{Tr}\pi'(hg) &= \delta_{\pi,\pi'} \frac{|G|}{\dim V_{\pi}} \operatorname{Tr}\pi(g), \\
\sum_{\pi\in\operatorname{Irr}(G)} \operatorname{Tr}\pi(g^{-1}) \operatorname{Tr}\pi(h) &= \delta_{\mathcal{C}_g,\mathcal{C}_h} |C_G(g)|
\end{aligned}
\]
for all $h,g\in G$ and $\pi, \pi' \in\operatorname{Irr}(G)$, where $\mathcal{C}_g$ denotes the conjugacy class containing $g$.

Here, we define the projection element $\chi_\mathcal{C}$ in $\mathbb{C}(G)$ as the characteristic function of $\mathcal{C}$. For any site $s$,
\[
P_{\chi_\mathcal{C}}(s) = \chi_\mathcal{C} \triangleright_{s}.
\]
For more general cases, we have
\[
P_{\mathcal{C},\pi} = \sum_{c\in \mathcal{C}} \delta_c \otimes q_c P_{\pi} q_c^{-1} = \frac{\dim V_{\pi}}{|C_G|} \sum_{c\in \mathcal{C}} \sum_{n\in C_G} \operatorname{Tr}\pi(n^{-1}) \delta_c \otimes q_c n q_c^{-1},
\]
where $P_{\pi} \in \mathbb{C}C_G$ is the projection operator for the irreducible representation $\pi$ of the centralizer group $C_G$. This defines the site projection operator
\[
P_{\mathcal{C},\pi}(s) = P_{\mathcal{C},\pi} \triangleright_{s}.
\]
Thus, for each site, we have an observable
\[
\Gamma_s = \sum_{\mathcal{C},\pi} r_{\mathcal{C},\pi} P_{\mathcal{C},\pi}(s).
\]

The projection operators for charge and flux excitations are given by
\[
\begin{aligned}
P_{e,1} &= \Lambda^* \otimes \Lambda, \\
P_{e,\pi} &= \Lambda^* \otimes P_{\pi}, \\
P_{\mathcal{C},1} &= \sum_{c\in \mathcal{C}} \delta_c \otimes q_c \Lambda_{C_G} q_c^{-1},
\end{aligned}
\]
where, for a fixed $\mathcal{C}$,
\[
\sum_{\pi\in\operatorname{Irr}(C_G)} P_{\mathcal{C},\pi} = \sum_{c\in \mathcal{C}} \delta_c \otimes q_c \left( \sum_{\pi\in\operatorname{Irr}(C_G)} P_\pi \right) q_c^{-1} = \chi_\mathcal{C} \otimes 1.
\]

A state $|\psi\rangle$ is said to be occupied by a quasiparticle of type $(\mathcal{C}, \pi)$ at a site if it is stabilized under the projection operator $P_{\mathcal{C},\pi}$, i.e.,
\[
P_{\mathcal{C},\pi} |\psi\rangle=|\psi\rangle.
\]

\begin{lemma}\cite{Maj04}
\label{1}
The operators $\{P_{\mathcal{C},\pi}\}$ form a complete orthogonal set of central projectors:
   \[
P_{\mathcal{C},\pi} P_{\mathcal{C}',\pi'} = \delta_{\mathcal{C},\mathcal{C}'} \delta_{\pi,\pi'} P_{\mathcal{C},\pi}
\]
   \[
\sum_{\mathcal{C},\pi} P_{\mathcal{C},\pi} = 1
\]

\end{lemma}

For the complex group algebra $\mathbb{C}G$, according to Wedderburn's theorem, there exists a algebra isomorphism $\mathbb{C}G \cong \bigoplus_{\pi} \text{End}(V_{\pi})$. The isomorphism maps an element $g$ to $\pi(g)_{ij}e_i \otimes f^j$, where $e_i$ is a basis of $V_{\pi}$, $f^j$ is the dual basis, and $e_i \otimes f^j$ corresponds to the elementary matrix $E_{ij}$ with 1 in the $i$-th row and $j$-th column when viewing $\text{End}(V_{\pi})$ as matrices. If $v = v^i e_i$, then $\pi(g)v = v^i \pi_{kj} e_k \langle f^j, e_i \rangle = e_k \pi_{ki} v^i$.

Conversely, we define
\[
\begin{aligned}
&\Phi_{\mathbb{C}G}: \bigoplus_{\pi} \text{End}(V_{\pi}) \to \mathbb{C}G, \\
&\Phi(e_i \otimes f^j) = \frac{\dim V_{\pi}}{|G|} \sum_{g \in G} \pi(g^{-1})_{ji} g,
\end{aligned}
\]
where $\Phi$ is a bimodule isomorphism. This can be interpreted as $\mathbb{C}G$ acting on itself or on $\text{End}(V_{\pi}) = V_{\pi} \otimes V_{\pi}^*$, with
\[
h \triangleright e_i = e_k \pi(h)_{ki} \quad \text{and} \quad f^j \triangleleft h = \pi(h)_{jk} f^k,
\]
and $\pi'(P_{\pi}) = \text{id} \, \delta_{\pi, \pi'}$. When $e_i \in V_{\pi}$ and $f^j \in V_{\pi}^*$, we have
\[
P_{\pi} \triangleright e_i = e_k \pi(P_{\pi})_{ki} = e_i, \quad f^j \triangleleft P_{\pi} = \pi(P_{\pi})_{jk} f^k = f^j,
\]
and zero otherwise. Through $\Phi$, the action of $P_{\pi}$ corresponds to left and right multiplication on $\mathbb{C}G$, giving
\[
P_{\pi} \mathbb{C}G = (\mathbb{C}G) P_{\pi} \cong \text{End}(V_{\pi}).
\]

Similarly, we let $D(G)$ act on $\text{End}(V_{\mathcal{C},\pi}) = V_{\mathcal{C},\pi} \otimes V^*_{\mathcal{C},\pi}$.

\begin{theorem} \cite{CM22}
Let $\{c \otimes e_i\}$ be a basis for the $D(G)$-representation $V_{\mathcal{C},\pi}$, and let $\{\delta_d \otimes f^j\}$ be the corresponding dual basis. Then the following map defines a bimodule isomorphism:
\[
\Phi: \bigoplus_{\mathcal{C}, \pi} \text{End}(V_{\mathcal{C}, \pi}) \to D(G),
\]
\[
\Phi(c \otimes e_i \otimes \delta_d \otimes f^j) = \delta_c \otimes q_c \Phi_{\mathbb{C}C_G}(e_i \otimes f^j) q_d^{-1} = \frac{\dim V_{\pi}}{|C_G|} \sum_{n \in C_G} \pi(n^{-1})_{ji} \delta_c \otimes q_c n q_d^{-1},
\]
where $C_G$ denotes the centralizer subgroup.

\end{theorem}

\subsection{Ribbon operator}
We define a ribbon as a strip-like path connecting two sites, composed of alternating edges and dual edges (plaquettes). If the endpoints of a ribbon lie at the same site, we call it a closed ribbon; if the endpoints lie at two distinct, disjoint sites, we call it an open ribbon. Note that there exist ribbons that are neither open nor closed---for example, a ribbon whose endpoints belong to different sites located on distinct faces but sharing the same vertex.

When considering only the vertex path or the face path individually, braiding (i.e., exchanging the positions of two quasiparticles) produces only a phase factor. Non-trivial results arise only when both paths are combined, which motivates the introduction of the ribbon operator $F_{\xi}^{h,g}$.

The figure below illustrates the action of the ribbon operator $F_{\xi}^{h,g}$. In the diagram, edges along the vertex path of the ribbon remain unchanged, while edges along the dashed line, rotated counterclockwise relative to their orientation, are acted upon by $h$ via conjugation along the traversed vertex path.

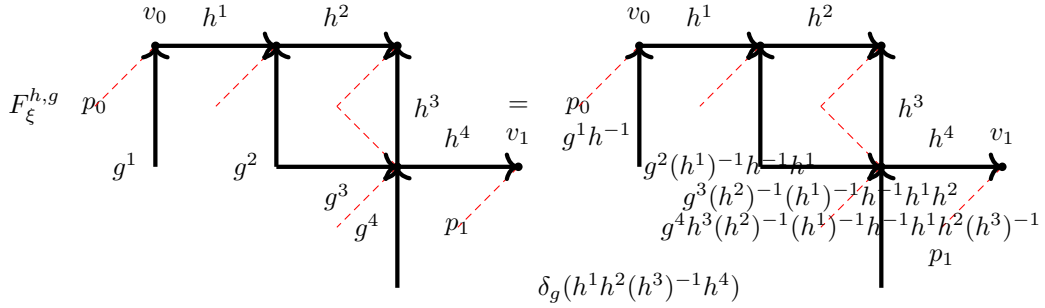
\begin{tikzpicture}[scale=0.8]
    \node (v0) at (0,0.5) {$v_0$};
    \node (v1) at (6,-1.5) {$v_1$};
    \node (p0) at (-1,-1) {$p_0$};
    \node (p1) at (5,-3) {$p_1$};
        \node (v0) at (0+8,0.5) {$v_0$};
    \node (v1) at (6+8,-1.5) {$v_1$};
    \node (p0) at (-1+8,-1) {$p_0$};
    \node (p1) at (5+8,-3.5) {$p_1$};
    \draw[densely dashed,red] (0,0) -- (-1,-1);
    \draw[densely dashed,red] (2,0) -- (1,-1);
    \draw[densely dashed,red] (4,0) -- (3,-1);
    \draw[densely dashed,red] (4,-2) -- (3,-1);
     \draw[densely dashed,red] (4,-2) -- (3,-3);
     \draw[densely dashed,red] (6,-2) -- (5,-3);
    \draw[densely dashed,red] (0+8,0) -- (-1+8,-1);
    \draw[densely dashed,red] (2+8,0) -- (1+8,-1);
    \draw[densely dashed,red] (4+8,0) -- (3+8,-1);
    \draw[densely dashed,red] (4+8,-2) -- (3+8,-1);
     \draw[densely dashed,red] (4+8,-2) -- (3+8,-3);
     \draw[densely dashed,red] (6+8,-2) -- (5+8,-3);
    \draw [->,ultra thick](0,0) -- (2,0);
    \draw [<-,ultra thick](0,0) -- (0,-2);
    \draw [<-,ultra thick](2,0) -- (2,-2);
    \draw [->,ultra thick](2,0) -- (4,0);
    \draw [->,ultra thick](2,-2) -- (4,-2);
    \draw [<-,ultra thick](4,0) -- (4,-2);
    \draw [->,ultra thick](4,-2) -- (6,-2);
    \draw [<-,ultra thick](4,-2) -- (4,-4);
    \draw [->,ultra thick](0+8,0) -- (2+8,0);
    \draw [<-,ultra thick](0+8,0) -- (0+8,-2);
    \draw [<-,ultra thick](2+8,0) -- (2+8,-2);
    \draw [->,ultra thick](2+8,0) -- (4+8,0);
    \draw [->,ultra thick](2+8,-2) -- (4+8,-2);
    \draw [<-,ultra thick](4+8,0) -- (4+8,-2);
    \draw [->,ultra thick](4+8,-2) -- (6+8,-2);
    \draw [<-,ultra thick](4+8,-2) -- (4+8,-4);

    \node at (-0.5, -2) {$g^1$};
    \node at (1.5, -2) {$g^2$};
    \node at (3, -2.5) {$g^3$};
    \node at (3.5, -3) {$g^4$};

    \node at (1, 0.5) {$h^1$};
    \node at (3, 0.5) {$h^2$};
    \node at (4.5, -1) {$h^3$};
    \node at (5, -1.5) {$h^4$};

    \node at (-2, -1) {$F_{\xi}^{h, g}$};
    \node at (6, -1) {=};
    \node at (-0.5+8-0.2, -1.5) {$g^1h^{-1}$};
    \node at (1.5+8, -2) {$g^2(h^1)^{-1}h^{-1}h^1$};
    \node at (3+8, -2.5) {$g^3(h^2)^{-1}(h^1)^{-1}h^{-1}h^1h^2$};
    \node at (3.5+8, -3) {$g^4h^3(h^2)^{-1}(h^1)^{-1}h^{-1}h^1h^2(h^3)^{-1}$};

    \node at (1+8, 0.5) {$h^1$};
    \node at (3+8, 0.5) {$h^2$};
    \node at (4.5+8, -1) {$h^3$};
    \node at (5+8, -1.5) {$h^4$};

    \node at (8, -4) {$\delta_g(h^1h^2(h^3)^{-1}h^4)$};

    \fill (0,0) circle (2pt);
    \fill (2,0) circle (2pt);
    \fill (4,0) circle (2pt);
    \fill (4,-2) circle (2pt);
	\fill (6,-2) circle (2pt);
    \fill (0+8,0) circle (2pt);
    \fill (2+8,0) circle (2pt);
    \fill (4+8,0) circle (2pt);
    \fill (4+8,-2) circle (2pt);
	\fill (6+8,-2) circle (2pt);
\end{tikzpicture}

By definition, we can compute the composition effects under the concatenation of ribbons of the same type connecting distinct sites, as well as the sequential application of distinct ribbons at the same site:
\[
F_{\xi' \circ \xi}^{h, g} = \sum_{f \in G} F_{\xi'}^{f^{-1} h f, f^{-1} g} \circ F_{\xi}^{h, f}; \quad F_{\xi}^{h, g} \circ F_{\xi'}^{h', g'} = \delta_{g, g'} F_{\xi}^{h h', g}.
\]
In the first formula, the upper indices involve a conjugation action on the first component, while the second component corresponds to the comultiplication $\Delta(\delta_g)$. From the second formula, we deduce that the adjoint of the ribbon operator is
\[
\left(F_{\xi}^{h, g}\right)^{\dagger} = F_{\xi}^{h^{-1}, g}.
\]

For a closed ribbon operator, we consider the action of its two components: the first component applies $\delta_{g^{-1}} \triangleright$ at the endpoints, whereas the second component sequentially applies conjugation actions along the vertex path associated with the ribbon. Through calculation, we observe that the edges connected to intermediate vertices along the ribbon remain invariant under the action of the ribbon operator, whereas only the edges at the endpoints exhibit non-trivial effects.

This leads us to analyze the commutation relations between ribbon operators, vertex operators, and face operators.

\begin{lemma}\cite{CM22} \label{5}
Let $\xi$ be a ribbon connecting sites $s_{0} = (v_{0}, p_{0})$ and $s_{1} = (v_{1}, p_{1})$. For all vertices $v \notin \{v_{0}, v_{1}\}$ and faces $p \notin \{p_{0}, p_{1}\}$, we have
\[
\left[F_{\xi}^{h, g}, f \triangleright_{v}\right] = 0, \quad \left[F_{\xi}^{h, g}, \delta_{e} \triangleright_{p}\right] = 0.
\]
For disjoint $s_{0}$ and $s_{1}$ (i.e., sharing neither vertices nor faces), the following relations hold:
\[
\begin{aligned}
& f \triangleright_{s_{0}} \circ F_{\xi}^{h, g} = F_{\xi}^{fhf^{-1}, fg} \circ f \triangleright_{s_{0}}, \quad \delta_{f} \triangleright_{s_{0}} \circ F_{\xi}^{h, g} = F_{\xi}^{h, g} \circ \delta_{h^{-1} f} \triangleright_{s_{0}}, \\
& f \triangleright_{s_{1}} \circ F_{\xi}^{h, g} = F_{\xi}^{h, gf^{-1}} \circ f \triangleright_{s_{1}}, \quad \delta_{f} \triangleright_{s_{1}} \circ F_{\xi}^{h, g} = F_{\xi}^{h, g} \circ \delta_{fg^{-1} hg} \triangleright_{s_{1}}.
\end{aligned}
\]
\end{lemma}

It follows from the preceding lemma that every term in the Hamiltonian whose support does not contain $s_0$ or $s_1$ commutes with the ribbon operator. At the endpoints $s_0$ and $s_1$, however, the nontrivial commutation relations derived above imply the creation of a quasiparticle at $s_0$ together with its corresponding antiparticle at $s_1$. Thus, the ribbon operator acts as a quasiparticle--antiquasiparticle pair-creation operator.

We next introduce triangle operators to make the local structure of ribbon operators explicit. A ribbon can be decomposed into a sequence of triangles, with the corresponding triangle operators serving as the elementary building blocks of the ribbon operator.

\begin{definition} (Triangular Operators)
To decompose ribbon operators into elementary units, we define
direct-triangle operator: $T_{\tau}^g$, acting along a directed triangle $\tau$.
dual-triangle operator: $L_{\tau^*}^h$, acting on a dual triangle $\tau^*$.

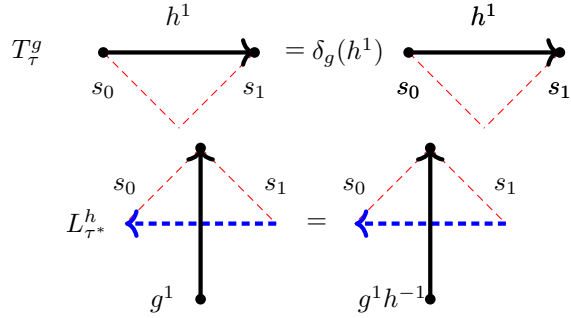
\[
\begin{tikzpicture}[xshift=4]

    \draw[densely dashed,red] (0,0) -- (1,-1);
    \draw[densely dashed,red] (2,0) -- (1,-1);

    \draw [->,ultra thick](0,0) -- (2,0);


    \node at (-1, 0) {$T_{\tau}^g$};
    \node at (1, 0.5) {$h^1$};
    \node at (0, -0.5) {$s_0$};
		\node at (2, -0.5) {$s_1$};

    \fill (0,0) circle (2pt);
    \fill (2,0) circle (2pt);

\end{tikzpicture}
\begin{tikzpicture}[xshift=7]

    \draw[densely dashed,red] (0,0) -- (1,-1);
    \draw[densely dashed,red] (2,0) -- (1,-1);

    \draw [->,ultra thick](0,0) -- (2,0);


    \node at (1, 0.5) {$h^1$};
    \node at (0, -0.5) {$s_0$};
		\node at (2, -0.5) {$s_1$};

    \fill (0,0) circle (2pt);
    \fill (2,0) circle (2pt);

		\node at (-1.5, 0) {$=$};
    \node at (-0.8, 0) {$\delta_{g}(h^1)$};
    \node at (1, 0.5) {$h^1$};
    \node at (0, -0.5) {$s_0$};
		\node at (2, -0.5) {$s_1$};
\end{tikzpicture}
\]
\[
\begin{tikzpicture}

    \draw[densely dashed,red] (0,0) -- (1,-1);
    \draw[densely dashed,red] (0,0) -- (-1,-1);

    \draw[->,ultra thick,densely dashed,blue] (1,-1) -- (-1,-1);

    \draw [<-,ultra thick](0,0) -- (0,-2);

    \node at (-0.5, -2) {$g^1$};
    \node at (-1.5, -1) {$L_{\tau^*}^h$};

    \node at (0-1, -0.5) {$s_0$};
		\node at (2-1, -0.5) {$s_1$};

    \fill (0,0) circle (2pt);
    \fill (0,-2) circle (2pt);
\end{tikzpicture}
\begin{tikzpicture}[xshift=3]

    \draw[densely dashed,red] (0,0) -- (1,-1);
    \draw[densely dashed,red] (0,0) -- (-1,-1);

    \draw[->,ultra thick,densely dashed,blue] (1,-1) -- (-1,-1);

    \draw [<-,ultra thick](0,0) -- (0,-2);

    \node at (-0.5, -2) {$g^1h^{-1}$};
    \node at (-1.5, -1) {$=$};

    \node at (0-1, -0.5) {$s_0$};
		\node at (2-1, -0.5) {$s_1$};

    \fill (0,0) circle (2pt);
    \fill (0,-2) circle (2pt);

\end{tikzpicture}
\]
\end{definition}

Using the orientation of the original lattice, together with the induced orientation of the dual lattice, we can decompose a ribbon into a sequence of triangles and express the corresponding ribbon operator in terms of triangle operators.

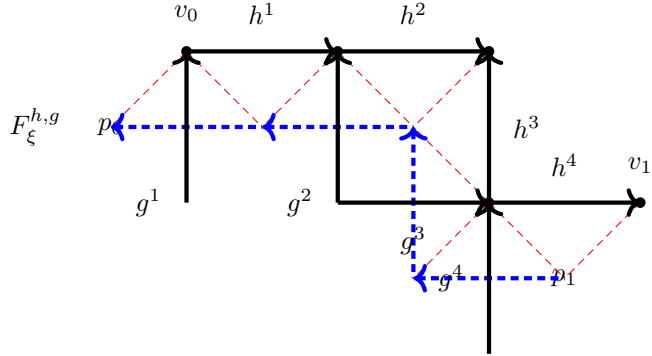
\[
\begin{tikzpicture}
    \node (v0) at (0,0.5) {$v_0$};
    \node (v1) at (6,-1.5) {$v_1$};
    \node (p0) at (-1,-1) {$p_0$};
    \node (p1) at (5,-3) {$p_1$};

    \draw[densely dashed,red] (0,0) -- (-1,-1);
    \draw[densely dashed,red] (0,0) -- (1,-1);
    \draw[densely dashed,red] (2,0) -- (1,-1);
    \draw[densely dashed,red] (2,0) -- (3,-1);
    \draw[densely dashed,red] (4,0) -- (3,-1);
    \draw[densely dashed,red] (4,-2) -- (3,-1);
    \draw[densely dashed,red] (4,-2) -- (5,-3);
    \draw[densely dashed,red] (4,-2) -- (3,-3);
    \draw[densely dashed,red] (6,-2) -- (5,-3);
    \draw[<-,ultra thick,densely dashed,blue] (-1,-1) -- (1,-1);
    \draw[<-,ultra thick,densely dashed,blue] (1,-1) -- (3,-1);
    \draw[<-,ultra thick,densely dashed,blue] (3,-1) -- (3,-3);
    \draw[<-,ultra thick,densely dashed,blue] (3,-3) -- (5,-3);

    \draw [->,ultra thick](0,0) -- (2,0);
    \draw [<-,ultra thick](0,0) -- (0,-2);
    \draw [<-,ultra thick](2,0) -- (2,-2);
    \draw [->,ultra thick](2,0) -- (4,0);
    \draw [->,ultra thick](2,-2) -- (4,-2);
    \draw [<-,ultra thick](4,0) -- (4,-2);
    \draw [->,ultra thick](4,-2) -- (6,-2);
    \draw [<-,ultra thick](4,-2) -- (4,-4);

    \node at (-0.5, -2) {$g^1$};
    \node at (1.5, -2) {$g^2$};
    \node at (3, -2.5) {$g^3$};
    \node at (3.5, -3) {$g^4$};

    \node at (1, 0.5) {$h^1$};
    \node at (3, 0.5) {$h^2$};
    \node at (4.5, -1) {$h^3$};
    \node at (5, -1.5) {$h^4$};

    \node at (-2, -1) {$F_{\xi}^{h, g}$};

    \fill (0,0) circle (2pt);
    \fill (2,0) circle (2pt);
    \fill (4,0) circle (2pt);
    \fill (4,-2) circle (2pt);
	\fill (6,-2) circle (2pt);
\end{tikzpicture}
\]
Furthermore, by definition,
\[
F_{\tau}^{h,g}=T_{\tau}^{g},
\qquad
F_{\tau^*}^{h,g}=\delta_{g,e}L_{\tau^*}^{h}.
\]
A general ribbon operator can be constructed by composing the corresponding triangle operators. The triangle operators satisfy the multiplication rules
\[
\begin{aligned}
T_{\tau}^{g}\circ T_{\tau}^{g'}
&=
\delta_{g,g'}T_{\tau}^{g},
\\
L_{\tau^*}^{h}\circ L_{\tau^*}^{h'}
&=
L_{\tau^*}^{hh'}.
\end{aligned}
\]
Consequently, the algebras generated by the direct and dual triangle operators are
\[
\begin{aligned}
\mathcal{A}_{\tau}
&:=
\operatorname{span}\{T_{\tau}^{g}\mid g\in G\}
\cong \mathbb{C}(G),
\\
\mathcal{A}_{\tau^*}
&:=
\operatorname{span}\{L_{\tau^*}^{h}\mid h\in G\}
\cong \mathbb{C}G.
\end{aligned}
\]
We next examine the action of ribbon operators on the ground state.

\begin{proposition} \cite{CM22}
Let \( |\text{vac}\rangle \) be the ground state on the plane \( \Sigma \), and let \( \xi \) be a ribbon connecting the specified sites \( s_{0} = (v_{0}, p_{0}) \) and \( s_{1} = (v_{1}, p_{1}) \). Define
\[
|\psi^{h, g}\rangle := F_{\xi}^{h, g}|\text{vac}\rangle.
\]

(1) The state \( |\psi^{h, g}\rangle \) depends only on the sites \( s_{0} \) and \( s_{1} \), and is independent of the choice of ribbon \( \xi \).

(2) The subspace
\[
\mathcal{L}(s_{0}, s_{1}) := \left\{ |\psi\rangle \in \mathcal{H} \, \big| \, A(v)|\psi\rangle = B(p)|\psi\rangle = |\psi\rangle, \, \forall v \notin \{v_{0}, v_{1}\}, \, p \notin \{p_{0}, p_{1}\} \right\}
\]
is spanned by \( \{|\psi^{h, g}\rangle \, | \, h, g \in G\} \).

(3) When the sites \( s_{0} = (v_{0}, p_{0}) \) and \( s_{1} = (v_{1}, p_{1}) \) are disjoint (i.e., the ribbon is open), the set \( \{|\psi^{h, g}\rangle \, | \, h, g \in G\} \) forms an orthonormal basis for \( \mathcal{L}(s_{0}, s_{1}) \), called the group basis.

(4) For disjoint sites \( s_{0} \) and \( s_{1} \), the action of operators on \( \mathcal{L}(s_{0}, s_{1}) \) follows \ref{5}:
\[
\begin{aligned}
& f \triangleright_{s_{0}} |\psi^{h, g}\rangle = |\psi^{fhf^{-1}, fg}\rangle, \quad \delta_{f} \triangleright_{s_{0}} |\psi^{h, g}\rangle = \delta_{f, h} |\psi^{h, g}\rangle, \\
& f \triangleright_{s_{1}} |\psi^{h, g}\rangle = |\psi^{h, gf^{-1}}\rangle, \quad \delta_{f} \triangleright_{s_{1}} |\psi^{h, g}\rangle = \delta_{f, g^{-1}h^{-1}g} |\psi^{h, g}\rangle.
\end{aligned}
\]
Via the mapping \( |\psi^{h, g}\rangle \mapsto \delta_{h} g \), this action is isomorphic to the left-right regular representation of \( D(G) \).
\end{proposition}

Part~(4) of Proposition~2.6 implies that $\mathcal{L}(s_0,s_1)$ admits an alternative basis consisting of quasiparticle excitation states. The following corollaries are immediate consequences.
\begin{corollary}\cite{CM22}
Let $\xi$ be an open ribbon connecting disjoint sites $s_0$ and $s_1$. Then
   $\mathcal{L}(s_0, s_1)$ admits a quasiparticle basis indexed by irreducible representations $(\mathcal{C}, \pi)$ of $D(G)$:
   \[
|u, v; \mathcal{C}, \pi\rangle = \frac{\dim V_\pi}{|C_G|} {F'}_{\xi}^{\mathcal{C},\pi; u,v} | \text{vac} \rangle,
\]
   where $u = (c, i)$, $v = (d, j)$, $c, d \in \mathcal{C}$, and ${F'}_{\xi}^{\mathcal{C}, \pi; u, v} := \sum_{n \in C_G} \pi(n^{-1})_{ji} F_{\xi}^{c, q_{c} n q_d^{-1}}$.
\end{corollary}

These states realize the left site action $\triangleright_{s_0}$ on $\mathcal{L}(s_0,s_1)$, while the action $\triangleright_{s_1}$ at the right endpoint is related to a right site action through the antipode $S$ of $D(G)$. The ribbon operators ${F'}_{\xi}^{\mathcal{C},\pi;u,v}$ form a basis of an operator space that commutes with all Hamiltonian terms except those supported at the ribbon endpoints.
\begin{corollary}\cite{CM22}
If a quantum state \( |\psi\rangle \in \mathcal{L}(s_{0}, s_{1}) \) can detect quasiparticles of type \( \mathcal{C},\pi \) via a non-zero projection \( P_{\mathcal{C},\pi} \triangleright_{s_{0}}|\psi\rangle \), then
\[
P_{\mathcal{C},\pi} \triangleright_{s_{0}}|\psi\rangle= |\psi\rangle\triangleleft_{s_{1}}P_{\mathcal{C},\pi},
\]
thus automatically detecting quasiparticles at \( s_{1} \), and vice versa. Specifically, for the Bell state
\[
| \text{Bell}, \xi \rangle = \sum_{h \in G} F_{\xi}^{h, e} | \text{vac} \rangle,
\]
all projections are non-zero, i.e., for any \( \mathcal{C},\pi \),
\[
P_{\mathcal{C},\pi} \triangleright_{s_{0}}|\text{Bell}, \xi\rangle= |\text{Bell}, \xi\rangle\triangleleft_{s_{1}}P_{\mathcal{C},\pi}.
\]
\end{corollary}

The space $\mathcal{L}(s_0,s_1)$ decomposes into distinct sectors labelled by the representations $(\mathcal{C},\pi)$, with basis states $|u,v;\mathcal{C},\pi\rangle$. The Bell state $|\text{Bell},\xi\rangle$ is given by the diagonal sum over these sectors:
\[
|\text{Bell},\xi\rangle
=
\sum_{\mathcal{C},\pi}\sum_u
|u,u;\mathcal{C},\pi\rangle.
\]
For a fixed representation $(\mathcal{C},\pi)$, projecting the Bell state onto the corresponding sector gives
\[
P_{\mathcal{C},\pi}\triangleright_{s_0}
|\text{Bell},\xi\rangle
=
|\text{Bell};\mathcal{C},\pi,\xi\rangle,
\]
where
\[
|\text{Bell};\mathcal{C},\pi,\xi\rangle
=
\sum_u |u,u;\mathcal{C},\pi\rangle.
\]
These states $|\text{Bell};\mathcal{C},\pi,\xi\rangle$ are referred to as minimal Bell states. They can be expressed explicitly in terms of ribbon operators as
\[
|\text{Bell};\mathcal{C},\pi,\xi\rangle
=
\frac{\dim V_\pi}{|C_G|}
W_\xi^{\mathcal{C},\pi}
|\text{vac}\rangle,
\qquad
W_\xi^{\mathcal{C},\pi}
:=
\sum_u {F'}_\xi^{\mathcal{C},\pi;u,u},
\]
where $W_\xi^{\mathcal{C},\pi}$ is called the ribbon trace operator.

For quasiparticle basis operators associated with open ribbons, the following relation holds:

\begin{lemma}  \cite{CM22}
Let \(\xi\) and \(\xi'\) be open ribbons connecting sites \( s_{0} = (v_{0}, p_{0}) \text{ to } s_{1} = (v_{1}, p_{1}) \) and \( s_{1} = (v_{1}, p_{1}) \text{ to } s_{2} = (v_{2}, p_{2}) \), respectively. Then
\[
{F'}_{\xi^{\prime} \circ \xi}^{\mathcal C, \pi; u, v} = \sum_{w} {F'}_{\xi^{\prime}}^{\mathcal C, \pi; w, v} \circ {F'}_{\xi}^{\mathcal C, \pi; u, w}.
\]
\end{lemma}
\begin{proof}
By the composition property of ribbon operators, we have
\[
\begin{aligned}
{F'}_{\xi'\circ\xi}^{\mathcal C,\pi;(c,i),(d,j)} &= \sum_{n\in C_G} \pi(n^{-1})_{ji} F_{\xi'\circ\xi}^{c,q_c n q_d^{-1}} \\
&= \sum_{f\in G} \sum_{n\in C_G} \pi(n^{-1})_{ji} F_{\xi'}^{f^{-1}c f, f^{-1} q_c n q_d^{-1}} \circ F_{\xi}^{c,f} \\
&= \sum_{b\in \mathcal C} \sum_{k} \sum_{m,n\in C_G} \pi((m^{-1} n)^{-1})_{jk} \pi(m^{-1})_{ki} F_{\xi'}^{b,q_b m^{-1} n q_d^{-1}} \circ F_{\xi}^{c,q_c m q_b^{-1}}\\
&= \sum_{w} {F'}_{\xi^{\prime}}^{\mathcal C, \pi; w, v} \circ {F'}_{\xi}^{\mathcal C, \pi; u, w},
\end{aligned}
\]
where \( w=(b,k) \).
\end{proof}

For ribbon trace operators defined on the same ribbon, the following relation holds in the electric charge sector.
\begin{lemma}  \cite{CM22}
\label{2}
\[
W_{\xi}^{e, \pi} \circ W_{\xi}^{e, \pi^{\prime}} = W_{\xi}^{e, \pi \otimes \pi^{\prime}}.
\]
\begin{proof}
Direct computation yields
\[
\begin{aligned}
W_{\xi}^{e,\pi} \circ W_{\xi}^{e,\pi'} & = \sum_{n,n'\in G} \operatorname{Tr}_{\pi}(n^{-1}) F_{\xi}^{e,n} \operatorname{Tr}_{\pi'}(n'^{-1}) F_{\xi}^{e,n'} \\
& = \sum_{n,n'\in G} \operatorname{Tr}_{\pi}(n^{-1}) \operatorname{Tr}_{\pi'}(n'^{-1}) \delta_{n,n'} F_{\xi}^{e,n} \\
& = \sum_{n} \operatorname{Tr}_{\pi\otimes\pi'}(n^{-1}) F_{\xi}^{e,n}.
\end{aligned}
\]
\end{proof}
\end{lemma}

For the adjoint of a ribbon trace operator $W_{\xi}^{\mathcal C,\pi}$, we have
\begin{lemma}\cite{CM22}\label{3}
Let \(\xi\) be an open ribbon connecting sites \( s_{0} = (v_{0}, p_{0}) \text{ to } s_{1} = (v_{1}, p_{1}) \). Then
\[
W_{\xi}^{\mathcal C, \pi \dagger} = W_{\xi}^{\mathcal C^{*}, \pi^{*}},
\]
where \( \pi^* \) is the conjugate unitary representation of \( C_G \), and \( \mathcal C^{*} = \mathcal C^{-1} \) satisfies
\[
r_{\mathcal C}^{*} = r_{\mathcal C}^{-1} \quad \text{and} \quad q_{c^{-1}} = q_{c} \quad \text{(defining \( q: \mathcal C^{-1} \rightarrow G \))}.
\]
\end{lemma}

Finally, we consider ribbons connecting multiple sites. The following lemma determines the dimension of the space generated by ribbon operators on $n+1$ sites and provides an orthogonal basis indexed by group elements.
\begin{lemma}\cite{CM22}\label{4}
In the \( D(G) \)-model on a boundaryless planar lattice \( \Sigma \), let \( s_0,s_1,...,s_n \) be \( n+1 \) disjoint sites. The space \( \mathcal{L}(s_{0}, s_{1}, \ldots, s_{n}) \) satisfies
\[
\dim (\mathcal{L}(s_{0}, s_{1}, \ldots, s_{n})) = |G|^{2n},
\]
with an orthogonal group basis:
\[
\{|\psi^{\{h^{1}, h^{2}, \ldots, h^{n}\},\{g^{1}, g^{2}, \ldots, g^{n}\}}\rangle \mid h^{1}, h^{2}, \ldots, h^{n}, g^{1}, g^{2}, \ldots, g^{n} \in G\}.
\]
\end{lemma}

It follows from this lemma that the subspaces generated by disjoint ribbons are mutually orthogonal.

\section{General ribbon-commutation relations and protected logical qudits for a finite group}\label{sec:construction}
Restricting to the sectors corresponding to electric excitations, we write $W_{\xi}^{\pi}:=W_{\xi}^{e,\pi}$.
Let
$
\varepsilon\in\operatorname{Irr}_1(G)
$
be a nontrivial one-dimensional irreducible representation of finite order
$
\operatorname{ord}(\varepsilon)=d>1.
$
Set
$
K:=\ker\varepsilon.
$
Then $K\triangleleft G$. Since the image of $\varepsilon$ is cyclic of order $d$, the first isomorphism theorem gives
$
G/K\cong \operatorname{Im}(\varepsilon)\cong \mathbb Z_d.
$
Choose $t\in G$ such that $tK$ generates $G/K$, and set
$
\omega:=\varepsilon(t).
$
After replacing $t$ by a suitable generator if necessary, we may assume that
$
\omega=e^{2\pi i/d}.
$
Hence
\begin{equation}
G=\bigsqcup_{r=0}^{d-1}t^rK,
\qquad
\varepsilon(t^rk)=\omega^r
\quad(k\in K).
\label{eq:coset-decomposition}
\end{equation}
Since $\varepsilon$ is one-dimensional, it factors through the abelianization of $G$, and therefore
$
[G,G]\subseteq K.
$

Moreover, if $G/[G,G]\cong \mathbb Z_d$ and $\varepsilon$ is faithful on the abelianization, then $K=[G,G]$.

\subsection{The representation-theoretic condition}
\begin{definition}
Let
$
\varepsilon\in\operatorname{Irr}_1(G)
$
be a one-dimensional irreducible representation. An irreducible representation
$
\rho\in\operatorname{Irr}(G)
$
is called an $\varepsilon$-stable representation if
$
\rho\otimes\varepsilon\simeq\rho.$
\end{definition}
For an $\varepsilon$-stable representation
$\rho$,
since $\operatorname{ord}(\varepsilon)=d$, this implies
\begin{equation*}
\rho\otimes\varepsilon^j\simeq\rho,\qquad\rho\otimes\varepsilon^{-j}\simeq\rho,
\qquad j=0,1,\ldots,d-1,
\label{eq:rho-stable-powers-fixed}
\end{equation*}

\begin{lemma}
Let $\rho\in\operatorname{Irr}(G)$. Then
$
\rho\in\operatorname{Irr}(G)
$
is  $\varepsilon$-stable
if and only if
\[
\operatorname{Tr}\rho(g)=0,
\qquad g\notin K=\ker\varepsilon.
\]
\end{lemma}
\begin{proof}
Since $\varepsilon$ is one-dimensional,
$
\operatorname{Tr}(\rho\otimes\varepsilon)(g)
=\varepsilon(g)\operatorname{Tr}\rho(g).
$
Thus $\rho\otimes\varepsilon\simeq\rho$ if and only if
$
(\varepsilon(g)-1)\operatorname{Tr}\rho(g)=0$
for all
$
 g\in G.
$
If $g\notin K$, then $\varepsilon(g)\neq1$, and hence $\operatorname{Tr}\rho(g)=0$.
Conversely, if $\operatorname{Tr}\rho(g)=0$ for every $g\notin K$, then the equality above holds for every $g\in G$; therefore $\rho\otimes\varepsilon$ and $\rho$ have the same character and are isomorphic.
\end{proof}

The central projection element corresponding to $\rho$ is
$
P_\rho
=
\frac{\dim V_\rho}{|G|}
\sum_{g\in G}
\bigl(\operatorname{Tr}\rho(g^{-1})\bigr)g.
$
For an $\varepsilon$-stable $\rho$, $\operatorname{Tr}\rho(g)=0$ for $g\notin K$, and therefore
\[
P_\rho
=
\frac{\dim V_\rho}{|G|}
\sum_{k\in K}
\bigl(\operatorname{Tr}\rho(k^{-1})\bigr)k.
\]
The corresponding electric charge projector is
$
P_{e,\rho}=\Lambda^*\otimes P_\rho=\delta_e\otimes P_\rho.
$
\subsection{Existence criterion for an \texorpdfstring{$\varepsilon$}{epsilon}-stable representation}
While the preceding lemma characterizes $\varepsilon$-stable irreducible representations, we now establish the conditions for their existence. Recall that $K=\ker\varepsilon$ and $G/K\simeq\mathbb Z_d$. Let $t\in G$ be an element such that the coset $tK$ generates $G/K$. For $\tau\in\operatorname{Irr}(K)$ and $g\in G$, we define the conjugate representation as\[
\tau^g(k):=\tau(g^{-1}kg), \qquad k\in K,
\]and its associated inertia subgroup as\[
I_G(\tau):=\{g\in G\mid \tau^g\simeq\tau\}.
\]

\begin{theorem}[Existence of an $\varepsilon$-stable representation]
There exists an $\varepsilon$-stable irreducible representation
$
\rho\in\operatorname{Irr}(G)
$
if and only if there exists
$
\tau\in\operatorname{Irr}(K)
$
such that
$
I_G(\tau)=K.
$
Equivalently, the action of $G/K\simeq\mathbb Z_d$ on
$\operatorname{Irr}(K)$ contains an orbit of length $d$; i.e.,
$
\tau,\tau^t,\tau^{t^2},\ldots,\tau^{t^{d-1}}
$
are pairwise inequivalent. In this case, one may take
$
\rho=\operatorname{Ind}_K^G\tau,
$
and
$
\dim V_\rho=d\,\dim V_\tau.
$
\end{theorem}

\begin{proof}
Suppose that $\tau\in\operatorname{Irr}(K)$ satisfies
$I_G(\tau)=K$, and set
$
\rho:=\operatorname{Ind}_K^G\tau.
$
Since $K\triangleleft G$ and $G/K=\langle tK\rangle$, the Mackey decomposition gives
\[
\operatorname{Res}_K^G\operatorname{Ind}_K^G\tau
\simeq
\bigoplus_{r=0}^{d-1}\tau^{t^r}.
\]
By Frobenius reciprocity,
\begin{align*}
\left\langle
\operatorname{Ind}_K^G\tau,
\operatorname{Ind}_K^G\tau
\right\rangle_G
&=
\left\langle
\tau,
\operatorname{Res}_K^G\operatorname{Ind}_K^G\tau
\right\rangle_K\\
&=
\sum_{r=0}^{d-1}\langle\tau,\tau^{t^r}\rangle_K
=1.
\end{align*}
Hence, $\rho\in\operatorname{Irr}(G)$. Since $K=\ker\varepsilon$, we have
$
\operatorname{Res}_K^G\varepsilon=1_K.
$
Using the tensor identity,
\begin{align*}
\rho\otimes\varepsilon
&=(\operatorname{Ind}_K^G\tau)\otimes\varepsilon\\
&\simeq
\operatorname{Ind}_K^G
\left(\tau\otimes\operatorname{Res}_K^G\varepsilon\right)\\
&=
\operatorname{Ind}_K^G\tau
=\rho.
\end{align*}
Thus, $\rho$ is $\varepsilon$-stable, and
$
\dim V_\rho=[G:K]\dim V_\tau=d\,\dim V_\tau.
$

Conversely, suppose that $\rho\in\operatorname{Irr}(G)$ is
$\varepsilon$-stable. Since $\operatorname{ord}(\varepsilon)=d$,
\[
\rho\otimes\varepsilon^j\simeq\rho,
\qquad j=0,1,\ldots,d-1.
\]
Because $G/K\simeq\mathbb Z_d$ and $\varepsilon$ generates its character group,
\[
\operatorname{Ind}_K^G1_K
\simeq
1\oplus\varepsilon\oplus\cdots\oplus\varepsilon^{d-1}.
\]
Hence,
\[
\operatorname{Ind}_K^G\operatorname{Res}_K^G\rho
\simeq
\rho\otimes\operatorname{Ind}_K^G1_K
\simeq
\rho^{\oplus d}.
\]
By Frobenius reciprocity,
\begin{equation}
\left\langle
\operatorname{Res}_K^G\rho,
\operatorname{Res}_K^G\rho
\right\rangle_K=d.
\label{eq:restriction-inner-product}
\end{equation}
Let $\tau\in\operatorname{Irr}(K)$ be an irreducible constituent of
$\operatorname{Res}_K^G\rho$, and set
$I:=I_G(\tau)$.
By Clifford theory,
\[
\operatorname{Res}_K^G\rho
\simeq
e\bigoplus_{x\in G/I}\tau^x
\]
for some positive integer $e$, where the representations $\tau^x$
are pairwise inequivalent. Hence
\[
\left\langle
\operatorname{Res}_K^G\rho,
\operatorname{Res}_K^G\rho
\right\rangle_K
=
e^2[G:I].
\]
Therefore,
$
e^2[G:I]=d.
$

Since $K\subseteq I\subseteq G$ and $G/K\simeq\mathbb Z_d$ is cyclic,
the quotient $I/K$ is cyclic. Hence $\tau$ extends to an irreducible
representation $\widetilde{\tau}$ of $I$. By Clifford correspondence,
there exists ${\sigma}\in\operatorname{Irr}(I)$ such that
$
\rho\simeq\operatorname{Ind}_I^G{\sigma}
$
and  $\operatorname{Res}_K^I{\sigma}\simeq e\tau$.
Since $\tau$ extends to $I$, every irreducible representation of $I$
lying above $\tau$ has the form
\[
{\sigma}\simeq\widetilde{\tau}\otimes\lambda,
\]
where $\lambda\in\operatorname{Irr}(I/K)$. Since $I/K$ is cyclic,
$\lambda$ is one-dimensional and is trivial on $K$. Thus
\[
\operatorname{Res}_K^I{\sigma}\simeq\tau,
\]
so the Clifford multiplicity is $e=1$.

Combining this with~\eqref{eq:restriction-inner-product} gives
\[
[G:I]=d=[G:K].
\]
Since $K\subseteq I$, we conclude that $I=K$.
\end{proof}

\subsection{Ribbon trace operators \texorpdfstring{$W_{\xi}^{\varepsilon^j}$}{W}}

Let $\xi$ be an open ribbon connecting the sites
$
s_0=(v_0,p_0)
$
and
$
s_1=(v_1,p_1).
$
For $j=0,1,\ldots,d-1$, the ribbon trace operator associated with the one-dimensional representation $\varepsilon^j$ is
\[
W_{\xi}^{\varepsilon^j}
=
\sum_{g\in G}
\varepsilon^j(g^{-1})F_{\xi}^{e,g}.
\]
Using~\eqref{eq:coset-decomposition}, we write $g=t^rk$ for $k\in K$. Then
$
\varepsilon^j(g^{-1})=\omega^{-jr},
$
and hence
\[
W_{\xi}^{\varepsilon^j}
=
\sum_{r=0}^{d-1}\omega^{-jr}
\sum_{k\in K}F_{\xi}^{e,t^rk}.
\]

\subsection{The relation between \texorpdfstring{$P_{e,\pi}$}{P} and \texorpdfstring{$W_{\xi}^{\varepsilon^j}$}{W}}
Let $\pi\in\operatorname{Irr}(G)$. Recall that
\[
P_\pi
=
\frac{\dim V_\pi}{|G|}
\sum_{g\in G}
\bigl(\operatorname{Tr}\pi(g^{-1})\bigr)g.
\]
According to \ref{5}, at the initial site $s_0$ of $\xi$,
\[
f\triangleright_{s_0}\circ F_{\xi}^{e,g}
=
F_{\xi}^{e,fg}\circ f\triangleright_{s_0}.
\label{eq:initial-basic-fixed}
\]
Therefore,
\begin{align*}
P_\pi\triangleright_{s_0}W_{\xi}^{\varepsilon^j}
&=
\frac{\dim V_\pi}{|G|}
\sum_{f,g\in G}
\bigl(\operatorname{Tr}\pi(f^{-1})\bigr)
\varepsilon^j(g^{-1})
\,f\triangleright_{s_0}\circ F_{\xi}^{e,g}
\notag\\
&=
\frac{\dim V_\pi}{|G|}
\sum_{f,g\in G}
\bigl(\operatorname{Tr}\pi(f^{-1})\bigr)
\varepsilon^j(g^{-1})
F_{\xi}^{e,fg}\circ f\triangleright_{s_0}.
\end{align*}
Setting $u=fg$, we have
\[
\varepsilon^j(g^{-1})
=
\varepsilon^j(u^{-1})\varepsilon^j(f).
\]
Hence,
\begin{align*}
P_\pi\triangleright_{s_0}W_{\xi}^{\varepsilon^j}
&=
\left(
\sum_{u\in G}\varepsilon^j(u^{-1})F_{\xi}^{e,u}
\right)
\circ
\left[
\frac{\dim V_\pi}{|G|}
\sum_{f\in G}
\varepsilon^j(f)
\bigl(\operatorname{Tr}\pi(f^{-1})\bigr)f
\right]\triangleright_{s_0}.
\end{align*}
The first factor is $W_{\xi}^{\varepsilon^j}$. Moreover, we note
\begin{align*}
\operatorname{Tr}(\pi\otimes\varepsilon^{-j})(f^{-1})
&=
\operatorname{Tr}\pi(f^{-1})\varepsilon^{-j}(f^{-1})
=
\varepsilon^j(f)\operatorname{Tr}\pi(f^{-1}).
\end{align*}
Thus, the second factor is $P_{\pi\otimes\varepsilon^{-j}}$, and therefore
\[
P_\pi\triangleright_{s_0}W_{\xi}^{\varepsilon^j}
=
W_{\xi}^{\varepsilon^j}\circ
P_{\pi\otimes\varepsilon^{-j}}\triangleright_{s_0}.
\]

At the terminal site $s_1$, \ref{5} gives
\[
f\triangleright_{s_1}\circ F_{\xi}^{e,g}
=
F_{\xi}^{e,gf^{-1}}\circ f\triangleright_{s_1}.
\]
Therefore,
\begin{align*}
P_\pi\triangleright_{s_1}W_{\xi}^{\varepsilon^j}
&=
\frac{\dim V_\pi}{|G|}
\sum_{f,g\in G}
\bigl(\operatorname{Tr}\pi(f^{-1})\bigr)
\varepsilon^j(g^{-1})
F_{\xi}^{e,gf^{-1}}\circ f\triangleright_{s_1}.
\end{align*}
Setting $u=gf^{-1}$, we find that
\[
\varepsilon^j(g^{-1})
=
\varepsilon^j(f^{-1})\varepsilon^j(u^{-1}).
\]
Consequently,
\begin{align*}
P_\pi\triangleright_{s_1}W_{\xi}^{\varepsilon^j}
&=
W_{\xi}^{\varepsilon^j}\circ
\left[
\frac{\dim V_\pi}{|G|}
\sum_{f\in G}
\varepsilon^j(f^{-1})
\bigl(\operatorname{Tr}\pi(f^{-1})\bigr)f
\right]\triangleright_{s_1}.
\end{align*}
Since
$
\operatorname{Tr}(\pi\otimes\varepsilon^j)(f^{-1})
=
\varepsilon^j(f^{-1})\operatorname{Tr}\pi(f^{-1}),
$
we obtain
\[
P_\pi\triangleright_{s_1}W_{\xi}^{\varepsilon^j}
=
W_{\xi}^{\varepsilon^j}\circ
P_{\pi\otimes\varepsilon^j}\triangleright_{s_1}.
\]

For the $\delta_e$ part, \ref{5} gives
\[
\delta_e\triangleright_{s_i}\circ F_{\xi}^{e,g}
=
F_{\xi}^{e,g}\circ\delta_e\triangleright_{s_i},
\quad i=0,1.
\]

These observations lead to the following lemma.

\begin{lemma}
Let $\pi\in\operatorname{Irr}(G)$, and let $\varepsilon\in\operatorname{Irr}_1(G)$ be a nontrivial one-dimensional irreducible representation. Then,\[
P_\pi\triangleright_{s_0}W_{\xi}^{\varepsilon^j}
=
W_{\xi}^{\varepsilon^j}\circ
P_{\pi\otimes\varepsilon^{-j}}\triangleright_{s_0},\quad
P_\pi\triangleright_{s_1}W_{\xi}^{\varepsilon^j}
=
W_{\xi}^{\varepsilon^j}\circ
P_{\pi\otimes\varepsilon^j}\triangleright_{s_1}.
\]
If $\pi=\rho$ is an $\varepsilon$-stable representation, we have
$
P_{\rho\otimes\varepsilon^{-j}}=P_{\rho\otimes\varepsilon^{j}}=P_\rho,
$
so
\[
P_\rho\triangleright_{s_i}W_{\xi}^{\varepsilon^j}
=
W_{\xi}^{\varepsilon^j}\circ P_\rho\triangleright_{s_i},\quad i=0,1.
\]
Then we obtain
\begin{equation}
P_{e,\rho}\triangleright_{s_i}W_{\xi}^{\varepsilon^j}
=
W_{\xi}^{\varepsilon^j}\circ
P_{e,\rho}\triangleright_{s_i},
\quad i=0,1.
\label{eq:full-projector-commutation-fixed}
\end{equation}
\end{lemma}

\subsection{Protected logical qudit}
We now provide a concrete construction of a protected logical qudit in the $D(G)$ Kitaev quantum double model. Assume the lattice on $\Sigma$ is in the ground state, and let $\xi$ be an open ribbon connecting the sites $s_0 = (v_0, p_0)$ and $s_1 = (v_1, p_1)$, as illustrated in the figure. Applying the ribbon operator $W_{\xi}^{e, \rho}$ generates a $\rho$ quasiparticle at $s_0$ and its antiparticle $\rho^*$ at $s_1$.

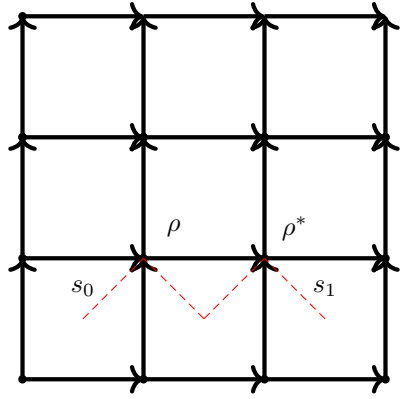
\[
\begin{tikzpicture}[scale=0.8]
    \draw (0,0) rectangle (6,6);
    \draw (0,2) -- (6,2);
        \draw (0,4) -- (6,4);
    \draw (2,0) -- (2,6);
        \draw (4,0) -- (4,6);
    \fill (0,0) circle (2pt);
    \fill (2,0) circle (2pt);
    \fill (4,0) circle (2pt);
    \fill (6,0) circle (2pt);
    \fill (0,2) circle (2pt);
    \fill (2,2) circle (2pt);
    \fill (4,2) circle (2pt);
    \fill (6,2) circle (2pt);
        \fill (0,4) circle (2pt);
        \fill (0,6) circle (2pt);
    \fill (2,4) circle (2pt);
    \fill (4,4) circle (2pt);
    \fill (6,4) circle (2pt);
        \fill (0,4) circle (2pt);
    \fill (2,4) circle (2pt);
    \fill (4,4) circle (2pt);
    \fill (6,4) circle (2pt);
    \draw [->,ultra thick](0,0) -- (2,0);
    \draw [->,ultra thick](2,0) -- (4,0);
    \draw [->,ultra thick](4,0) -- (6,0);
    \draw [->,ultra thick](0,2) -- (2,2);
    \draw [->,ultra thick](2,2) -- (4,2);
    \draw [->,ultra thick](4,2) -- (6,2);
    \draw [->,ultra thick](0,4) -- (2,4);
    \draw [->,ultra thick](2,4) -- (4,4);
    \draw [->,ultra thick](4,4) -- (6,4);
    \draw [->,ultra thick](0,6) -- (2,6);
    \draw [->,ultra thick](2,6) -- (4,6);
    \draw [->,ultra thick](4,6) -- (6,6);
    \draw [->,ultra thick](0,0) -- (0,2);
    \draw [->,ultra thick](0,2) -- (0,4);
    \draw [->,ultra thick](0,4) -- (0,6);
    \draw [->,ultra thick](2,0) -- (2,2);
    \draw [->,ultra thick](2,2) -- (2,4);
    \draw [->,ultra thick](2,4) -- (2,6);
    \draw [->,ultra thick](4,0) -- (4,2);
    \draw [->,ultra thick](4,2) -- (4,4);
    \draw [->,ultra thick](4,4) -- (4,6);
    \draw [->,ultra thick](6,0) -- (6,2);
    \draw [->,ultra thick](6,2) -- (6,4);
    \draw [->,ultra thick](6,4) -- (6,6);
    \draw[densely dashed,red] (1,1) -- (2,2);
    \draw[densely dashed,red] (3,1) -- (2,2);
    \draw[densely dashed,red] (3,1) -- (4,2);
    \draw[densely dashed,red] (5,1) -- (4,2);
     \node at (1, 1.5) {$s_0$};
     \node at (5, 1.5) {$s_1$};
     \node at (2.5, 2.5) {$\rho$};
     \node at (4.5, 2.5) {$\rho^*$};

\end{tikzpicture}
\]

Next, we apply the ribbon operator $W_{\xi'}^{\rho}$ along another open ribbon $\xi'$ connecting sites $s_2 = (v_2, p_2)$ and $s_3 = (v_3, p_3)$. We define the logical zero state as
\[
|0_L\rangle := W_{\xi'}^{\rho} \circ W_{\xi}^{\rho}|vac\rangle,
\]
in which the $\rho$ quasiparticles and their antiparticles $\rho^*$ occupy four distinct sites. Note that this state satisfies
\[
P_{e,\rho} \triangleright_{s_0} W_{\xi'}^{\rho} \circ W_{\xi}^{\rho} |vac\rangle = W_{\xi'}^{\rho} \circ W_{\xi}^{\rho} |vac\rangle.
\]

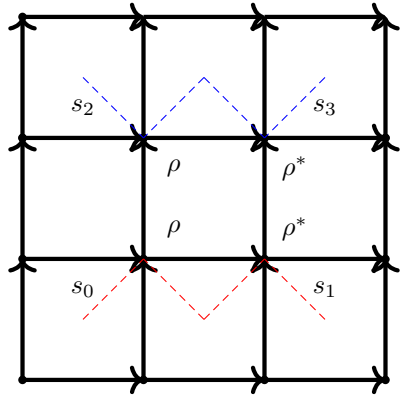
\[
\begin{tikzpicture}[scale=0.8]
    \draw (0,0) rectangle (6,6);
    \draw (0,2) -- (6,2);
        \draw (0,4) -- (6,4);
    \draw (2,0) -- (2,6);
        \draw (4,0) -- (4,6);
    \fill (0,0) circle (2pt);
    \fill (2,0) circle (2pt);
    \fill (4,0) circle (2pt);
    \fill (6,0) circle (2pt);
    \fill (0,2) circle (2pt);
    \fill (2,2) circle (2pt);
    \fill (4,2) circle (2pt);
    \fill (6,2) circle (2pt);
        \fill (0,4) circle (2pt);
                \fill (0,6) circle (2pt);
    \fill (2,4) circle (2pt);
    \fill (4,4) circle (2pt);
    \fill (6,4) circle (2pt);
        \fill (0,4) circle (2pt);
    \fill (2,4) circle (2pt);
    \fill (4,4) circle (2pt);
    \fill (6,4) circle (2pt);
    \draw [->,ultra thick](0,0) -- (2,0);
    \draw [->,ultra thick](2,0) -- (4,0);
    \draw [->,ultra thick](4,0) -- (6,0);
    \draw [->,ultra thick](0,2) -- (2,2);
    \draw [->,ultra thick](2,2) -- (4,2);
    \draw [->,ultra thick](4,2) -- (6,2);
    \draw [->,ultra thick](0,4) -- (2,4);
    \draw [->,ultra thick](2,4) -- (4,4);
    \draw [->,ultra thick](4,4) -- (6,4);
    \draw [->,ultra thick](0,6) -- (2,6);
    \draw [->,ultra thick](2,6) -- (4,6);
    \draw [->,ultra thick](4,6) -- (6,6);
    \draw [->,ultra thick](0,0) -- (0,2);
    \draw [->,ultra thick](0,2) -- (0,4);
    \draw [->,ultra thick](0,4) -- (0,6);
    \draw [->,ultra thick](2,0) -- (2,2);
    \draw [->,ultra thick](2,2) -- (2,4);
    \draw [->,ultra thick](2,4) -- (2,6);
    \draw [->,ultra thick](4,0) -- (4,2);
    \draw [->,ultra thick](4,2) -- (4,4);
    \draw [->,ultra thick](4,4) -- (4,6);
    \draw [->,ultra thick](6,0) -- (6,2);
    \draw [->,ultra thick](6,2) -- (6,4);
    \draw [->,ultra thick](6,4) -- (6,6);
    \draw[densely dashed,red] (1,1) -- (2,2);
    \draw[densely dashed,red] (3,1) -- (2,2);
    \draw[densely dashed,red] (3,1) -- (4,2);
    \draw[densely dashed,red] (5,1) -- (4,2);
     \node at (1, 1.5) {$s_0$};
     \node at (5, 1.5) {$s_1$};
     \node at (2.5, 2.5) {$\rho$};
     \node at (4.5, 2.5) {$\rho^*$};
    \draw[densely dashed,blue] (1,5) -- (2,4);
    \draw[densely dashed,blue] (3,5) -- (2,4);
    \draw[densely dashed,blue] (3,5) -- (4,4);
    \draw[densely dashed,blue] (5,5) -- (4,4);
     \node at (1, 4.5) {$s_2$};
     \node at (5, 4.5) {$s_3$};
     \node at (2.5, 3.5) {$\rho$};
     \node at (4.5, 3.5) {$\rho^*$};
\end{tikzpicture}
\]

Finally, we consider applying the ribbon operator $W_{\xi''}^{\varepsilon^j}$ along the open ribbon $\xi''$ connecting sites $s_0 = (v_0, p_0)$ and $s_2 = (v_2, p_2)$, and define
\[
|j_L\rangle
:=
W_{\xi''}^{\varepsilon^j}|0_L\rangle,
\qquad j=0,1,\ldots,d-1.
\]

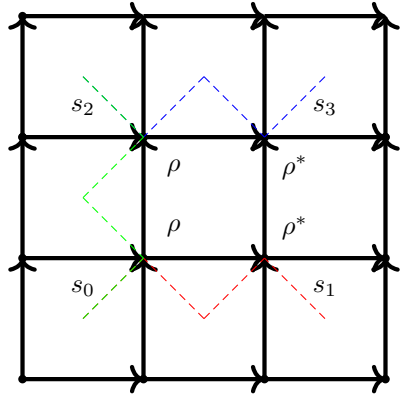
\[
\begin{tikzpicture}[scale=0.8]
    \draw (0,0) rectangle (6,6);
    \draw (0,2) -- (6,2);
        \draw (0,4) -- (6,4);
    \draw (2,0) -- (2,6);
        \draw (4,0) -- (4,6);
    \fill (0,0) circle (2pt);
    \fill (2,0) circle (2pt);
    \fill (4,0) circle (2pt);
    \fill (6,0) circle (2pt);
    \fill (0,2) circle (2pt);
    \fill (2,2) circle (2pt);
    \fill (4,2) circle (2pt);
    \fill (6,2) circle (2pt);
        \fill (0,4) circle (2pt);
                \fill (0,6) circle (2pt);
    \fill (2,4) circle (2pt);
    \fill (4,4) circle (2pt);
    \fill (6,4) circle (2pt);
        \fill (0,4) circle (2pt);
    \fill (2,4) circle (2pt);
    \fill (4,4) circle (2pt);
    \fill (6,4) circle (2pt);
    \draw [->,ultra thick](0,0) -- (2,0);
    \draw [->,ultra thick](2,0) -- (4,0);
    \draw [->,ultra thick](4,0) -- (6,0);
    \draw [->,ultra thick](0,2) -- (2,2);
    \draw [->,ultra thick](2,2) -- (4,2);
    \draw [->,ultra thick](4,2) -- (6,2);
    \draw [->,ultra thick](0,4) -- (2,4);
    \draw [->,ultra thick](2,4) -- (4,4);
    \draw [->,ultra thick](4,4) -- (6,4);
    \draw [->,ultra thick](0,6) -- (2,6);
    \draw [->,ultra thick](2,6) -- (4,6);
    \draw [->,ultra thick](4,6) -- (6,6);
    \draw [->,ultra thick](0,0) -- (0,2);
    \draw [->,ultra thick](0,2) -- (0,4);
    \draw [->,ultra thick](0,4) -- (0,6);
    \draw [->,ultra thick](2,0) -- (2,2);
    \draw [->,ultra thick](2,2) -- (2,4);
    \draw [->,ultra thick](2,4) -- (2,6);
    \draw [->,ultra thick](4,0) -- (4,2);
    \draw [->,ultra thick](4,2) -- (4,4);
    \draw [->,ultra thick](4,4) -- (4,6);
    \draw [->,ultra thick](6,0) -- (6,2);
    \draw [->,ultra thick](6,2) -- (6,4);
    \draw [->,ultra thick](6,4) -- (6,6);
    \draw[densely dashed,red] (1,1) -- (2,2);
    \draw[densely dashed,red] (3,1) -- (2,2);
    \draw[densely dashed,red] (3,1) -- (4,2);
    \draw[densely dashed,red] (5,1) -- (4,2);
     \node at (1, 1.5) {$s_0$};
     \node at (5, 1.5) {$s_1$};
     \node at (2.5, 2.5) {$\rho$};
     \node at (4.5, 2.5) {$\rho^*$};
    \draw[densely dashed,blue] (1,5) -- (2,4);
    \draw[densely dashed,blue] (3,5) -- (2,4);
    \draw[densely dashed,blue] (3,5) -- (4,4);
    \draw[densely dashed,blue] (5,5) -- (4,4);
     \node at (1, 4.5) {$s_2$};
     \node at (5, 4.5) {$s_3$};
     \node at (2.5, 3.5) {$\rho$};
     \node at (4.5, 3.5) {$\rho^*$};
    \draw[densely dashed,green] (1,3) -- (2,2);
    \draw[densely dashed,green] (1,3) -- (2,4);
        \draw[densely dashed,green] (1,1) -- (2,2);
            \draw[densely dashed,green] (1,5) -- (2,4);

\end{tikzpicture}
\]

Since
$
P_{e,\rho}\triangleright_{s_0}|0_L\rangle=|0_L\rangle,
$
using~\eqref{eq:full-projector-commutation-fixed} gives
\begin{align*}
P_{e,\rho}\triangleright_{s_0}|j_L\rangle
&=
P_{e,\rho}\triangleright_{s_0}W_{\xi''}^{\varepsilon^j}|0_L\rangle
=
W_{\xi''}^{\varepsilon^j}\circ
P_{e,\rho}\triangleright_{s_0}|0_L\rangle \notag\\
&=
W_{\xi''}^{\varepsilon^j}|0_L\rangle
=|j_L\rangle.
\end{align*}
Similarly, we have
$
P_{e,\rho}\triangleright_{s_2}|j_L\rangle=|j_L\rangle.
$

Since $\xi''$ has no endpoint at $s_1$ or $s_3$, \ref{5} gives
\[
P_{e,\rho^*}\triangleright_{s_1}|j_L\rangle=|j_L\rangle,
\qquad
P_{e,\rho^*}\triangleright_{s_3}|j_L\rangle=|j_L\rangle.
\]
Thus, all $|j_L\rangle$ occupy the same local quasiparticle sectors as $|0_L\rangle$.

By the orthogonality of the projection operators, for any
$
(\mathcal{C},\pi)\neq(e,\rho),
$
we have
\[
P_{\mathcal{C},\pi}\triangleright_{s_0}|j_L\rangle
=
P_{\mathcal{C},\pi}\triangleright_{s_2}|j_L\rangle
=0.
\]

\begin{lemma}
For any $(\mathcal{C}, \pi) \neq (e,\rho)$, $P_{\mathcal{C},\pi}\triangleright_{s_0}|j_L\rangle = P_{\mathcal{C},\pi}\triangleright_{s_2}|j_L\rangle = 0$, and $P_{e,\rho}\triangleright_{s_0}|j_L\rangle = P_{e,\rho}\triangleright_{s_2}|j_L\rangle =|j_L\rangle$.
\end{lemma}

At $s_1$ and $s_3$, the analogous statement is obtained by replacing $\rho$ by $\rho^*$.

We identify
$
X_L:=W_{\xi''}^{\varepsilon}.
$
By \ref{2}, $W_{\xi''}^{\varepsilon}\circ W_{\xi''}^{\varepsilon^j} = W_{\xi''}^{\varepsilon^{j+1}}$, and hence
\[
X_L|j_L\rangle =|((j+1)\bmod d)_L\rangle.
\]
Moreover,
$
X_L^d
=
W_{\xi''}^{\varepsilon^d}
=
W_{\xi''}^{1}
=Id.
$
Next, we verify the orthogonality of $|j_L\rangle$ and $|k_L\rangle$:
\begin{align*}
 \langle j_L | k_L \rangle &= \langle \mathrm{vac} | W_{\xi}^{\rho^\dagger} \circ W_{\xi'}^{\rho^\dagger} \circ W_{\xi''}^{{\varepsilon^{j}}^\dagger}\circ W_{\xi''}^{\varepsilon^{k}}  \circ W_{\xi'}^{\rho} \circ W_{\xi}^{\rho} | \mathrm{vac} \rangle \notag\\
& = \langle \mathrm{vac} | W_{\xi''}^{{\varepsilon}^{k-j}} \circ  W_{\xi'}^{\rho^\dagger} \circ W_{\xi'}^{\rho} \circ W_{\xi}^{\rho^\dagger} \circ W_{\xi}^{\rho} | \mathrm{vac} \rangle \notag\\
& = \langle \mathrm{vac} | W_{\xi''}^{{\varepsilon}^{k-j}}  \circ W_{\xi'}^{\rho^* \otimes \rho} \circ W_{\xi}^{\rho^* \otimes \rho} | \mathrm{vac} \rangle.
\end{align*}
This follows because for any ribbons $\xi, \xi'$ and $g, g' \in G$, we have
$
[F_{\xi}^{e,g}, F_{\xi^{\prime}}^{e,g^{\prime}}] = 0.
$
and the third equality is derived from \ref{2} and \ref{3}. Furthermore,  we note that $ W_{\xi''}^{{\varepsilon}^{j-k}}  | \mathrm{vac} \rangle \in \mathcal{L}(s_0, s_2) $, while $ W_{\xi'}^{\rho^* \otimes \rho} \circ W_{\xi}^{\rho^* \otimes \rho} | \mathrm{vac} \rangle \notin \mathcal{L}(s_0, s_2)$. Thus, by \ref{4}, we obtain $\langle j_L | k_L \rangle  = 0.$

Define
\[
\mathcal H_L:=\operatorname{span}\{|0_L\rangle,|1_L\rangle,\ldots,|(d-1)_L\rangle\}.
\]
Then $\dim\mathcal H_L=d$.

\begin{theorem}[Logical qudit and quasiparticle-sector detection]
Let $G$ be a finite group and let
$
\varepsilon\in\operatorname{Irr}_1(G)
$
with $\operatorname{ord}(\varepsilon)=d>1$. Suppose
$
\rho\in\operatorname{Irr}(G)
$
is $\varepsilon$-stable.
Then $\mathcal H_L$ is a $d$-dimensional protected logical subspace. Its logical $X$ gate is
$
X_L=W_{\xi''}^{\varepsilon},
$
with
\[
X_L|j_L\rangle=|((j+1)\bmod d)_L\rangle,
\qquad
X_L^d=Id.
\]
Moreover, for $i=0,2$,
\[
P_{e,\rho}\triangleright_{s_i}|j_L\rangle
=
|j_L\rangle,
\quad j=0,\ldots,d-1,
\]
whereas
\[
P_{\mathcal{C},\pi}\triangleright_{s_i}|j_L\rangle
=
0
\]
for every
$
(\mathcal{C},\pi)\neq(e,\rho).
$

Therefore, any error which takes the logical subspace from the $(e,\rho)$ sector into a different quasiparticle sector is detectable by the corresponding quasiparticle-sector measurement.
\end{theorem}

\subsection{Symmetric groups, alternating groups, and arbitrary logical qudits}

We now apply the preceding criterion to the symmetric and alternating groups,
and then construct an explicit semidirect-product family realizing arbitrary
logical qudits.

\begin{corollary}[Symmetric groups]
For every $n\geq3$, the group $S_n$ admits a
$\operatorname{sgn}$-stable irreducible representation. Consequently, the
$D(S_n)$ model admits the above logical-qubit construction with
\[
X_L=W_{\xi''}^{\operatorname{sgn}}, \qquad X_L^2=Id.
\]
\end{corollary}

\begin{proof}
For $n\geq2$,
\[
[S_n,S_n]=A_n, \qquad S_n/A_n\simeq\mathbb Z_2,
\]
so the only nontrivial one-dimensional irreducible representation is the sign
representation
\[
\varepsilon=\operatorname{sgn}, \qquad \operatorname{ord}(\operatorname{sgn})=2.
\]
The irreducible representations of $S_n$ are the Specht modules $S^\lambda$
indexed by partitions $\lambda\vdash n$, and
\[
S^\lambda\otimes\operatorname{sgn}\simeq S^{\lambda^t}.
\]
Thus, $S^\lambda$ is $\operatorname{sgn}$-stable exactly when
$\lambda=\lambda^t$. Such a self-conjugate partition exists for every
$n\geq3$. If $n=2m+1$ is odd, take
$
\lambda=(m+1,1^m),
$
and if $n=2m$ with $m\geq2$, take
$
\lambda=(m,2,1^{m-2}).
$
Both cases are self-conjugate. Hence, an $\operatorname{sgn}$-stable irreducible
representation exists, and the logical operator has order $2$.
\end{proof}

Next, we consider the alternating groups; in particular, we show that $A_4$ is the unique case realizing a logical qutrit.

\begin{corollary}[Alternating groups]
Among the alternating groups $A_n$ with $n\geq3$, $A_4$ is the unique group that supports a
nontrivial logical qudit by the present $\varepsilon$-stable-representation
mechanism. The group $A_4$ yields a logical qutrit, whereas $A_3$ and
$A_n$ for $n\geq5$ do not satisfy the required conditions.
\end{corollary}

\begin{proof}
For $A_4$, as detailed in the appendix, we have
\[
[A_4,A_4]=\{e,u_1,u_2,u_3\}=:U, \qquad A_4/U\simeq\mathbb Z_3.
\]
Take $\varepsilon=\chi_1$, with $\operatorname{ord}(\chi_1)=3$ and $K=\ker\chi_1=U$.
The three nontrivial irreducible representations of $U\simeq\mathbb Z_2\times\mathbb Z_2$ form a single orbit of length $3$ under the conjugation action of $A_4/U$. Hence, for any of the nontrivial irreducible representations $\tau\in\operatorname{Irr}(U)$, we have $I_{A_4}(\tau)=U=K$. The existence theorem therefore gives
\[
\rho=\operatorname{Ind}_{U}^{A_4}\tau, \qquad \dim V_\rho=3, \qquad \rho\otimes\chi_1\simeq\rho.
\]
This irreducible representation $\rho$ is precisely the representation $\chi_3$ discussed in the appendix, which is consistent with $\chi_1\otimes\chi_3\simeq\chi_3$. Consequently, we obtain $X_L=W_{\xi''}^{\chi_1}$ and $X_L^3=Id$, which admits a logical qutrit.

For $A_3\simeq\mathbb Z_3$, every irreducible representation is one-dimensional. If $\varepsilon\neq1$ and $\rho\otimes\varepsilon\simeq\rho$, the cancellation of one-dimensional characters would imply $\varepsilon=1$, which is a contradiction. Thus, $A_3$ admits no nontrivial $\varepsilon$-stable representation.

For $n\geq5$, the group $A_n$ is perfect:
\[
[A_n,A_n]=A_n.
\]
Therefore,
\[
A_n/[A_n,A_n]=1, \qquad \operatorname{Irr}_1(A_n)=\{1\}.
\]
There is no nontrivial one-dimensional representation $\varepsilon$, and the present construction cannot produce a nontrivial logical qudit.
\end{proof}
Furthermore, we provide a semidirect-product realization for arbitrary logical qudits.

\begin{theorem}[Semidirect-product realization of arbitrary logical qudits]
For every integer $d\geq2$, there exists a finite non-Abelian group $G_d$ with
a one-dimensional irreducible representation $\varepsilon$ of order $d$ and
an $\varepsilon$-stable irreducible representation $\rho_d$ of dimension
$d$. Consequently, the preceding construction realizes a logical
$d$-level qudit.
\end{theorem}

\begin{proof}
Let $K_d:=(\mathbb Z_2)^d$, and let $\mathbb Z_d=\langle t\rangle$
act on $K_d$ by cyclic permutation of the coordinates:
\[
t\triangleright(x_0,x_1,\ldots,x_{d-1})
:=t^{-1}(x_0,x_1,\ldots,x_{d-1})t
:=(x_1,x_2,\ldots,x_{d-1},x_0).
\]
Define
$
G_d:=K_d\rtimes\mathbb Z_d.
$
Let
$
\omega=e^{2\pi i/d}
$
and define
$
\varepsilon(x,t^r):=\omega^r.
$
Then, $\operatorname{ord}(\varepsilon)=d$, $\ker\varepsilon=K_d$, and $G_d/K_d\simeq\mathbb Z_d$.
Since $K_d$ is Abelian, all its irreducible representations are
one-dimensional. Define
$
\tau(x_0,x_1,\ldots,x_{d-1}):=(-1)^{x_0}.
$
For $r=0,\ldots,d-1$, conjugation by $t^r$ cyclically
permutes the coordinates of $K_d$, and hence
$
\tau^{t^r}(x_0,\ldots,x_{d-1})
=
(-1)^{x_r},
$
up to the choice of orientation of the cyclic action.
Thus, the representations
$
\tau,\tau^t,\ldots,\tau^{t^{d-1}}
$
are pairwise inequivalent.

Since $K_d$ is Abelian,
$
\tau^k=\tau
$
for every $k\in K_d$.
Every element of $G_d$ can be written uniquely as $kt^r$ with
$k\in K_d$ and $0\leq r<d$. Therefore,
$
\tau^{kt^r}\simeq\tau
$
holds if and only if $r=0$. Consequently,
we have $I_{G_d}(\tau)=K_d.$
By the existence theorem,
$
\rho_d:=\operatorname{Ind}_{K_d}^{G_d}\tau
$
is irreducible and $\varepsilon$-stable, and
$
\dim V_{\rho_d}
=[G_d:K_d]\dim V_\tau=d.
$

The logical states and logical $X$ gate are
$
|0_L\rangle := W_{\xi'}^{\rho_d} \circ W_{\xi}^{\rho_d}|\mathrm{vac}\rangle,
$
$
|j_L\rangle=W_{\xi''}^{\varepsilon^j}|0_L\rangle
$
for $j=0,1,\ldots,d-1$, and
$X_L=W_{\xi''}^{\varepsilon}.$
By Lemma~2.4,
$
X_L|j_L\rangle=|((j+1)\bmod d)_L\rangle,
$
and
$
X_L^d=Id.
$
Thus, the family
$
(\mathbb Z_2)^d\rtimes\mathbb Z_d
$
realizes logical qudits of arbitrary dimension $d$ in the present framework.
\end{proof}

\subsection{Universal quantum computation for \texorpdfstring{$D(A_4)$}{D(A4)}}
In this section, we present a scheme for universal quantum computation using ZX-calculus. The relevant material required for the ZX-calculus can be found in \cite{Ran14}.

It was shown in \cite{Luo11} that there exist transport operators $M_{\xi}^{\rho}$ that move $\rho$ quasiparticles deterministically along the lattice. We do not provide an explicit construction of $M_{\xi}^{\rho}$ here; however, assuming that such operators exist, they can be expressed as linear combinations of chargeon ribbon operators. In particular, a quasiparticle--antiquasiparticle pair connected by the ribbon $\xi$ can be fused back to the vacuum according to
$
M^\rho_{-\xi}W^\rho_\xi|\mathrm{vac}\rangle
=
W^\rho_{(-\xi)\circ\xi}|\mathrm{vac}\rangle
=
|\mathrm{vac}\rangle.
$
Therefore,
\begin{align*}
M^\rho_{-\xi}M^\rho_{-\xi'}|j_L\rangle
&=
M^\rho_{-\xi}M^\rho_{-\xi'}
W^{\varepsilon^j}_{\xi''}
W^\rho_{\xi'}
W^\rho_\xi
|\mathrm{vac}\rangle
\\
&=
W^{\varepsilon^j}_{\xi''}
M^\rho_{-\xi}M^\rho_{-\xi'}
W^\rho_{\xi'}
W^\rho_\xi
|\mathrm{vac}\rangle
\\
&=
W^{\varepsilon^j}_{\xi''}|\mathrm{vac}\rangle.
\end{align*}
Thus, after fusing the quasiparticle--antiquasiparticle pairs back to the vacuum, the logical information is carried by the remaining $\varepsilon^j$ excitation. Measuring the projector $P_{e,\varepsilon^j}\triangleright_{s_0}$ then allows the corresponding logical state to be distinguished.

Given two logical qutrits labeled $a$ and $b$, we have
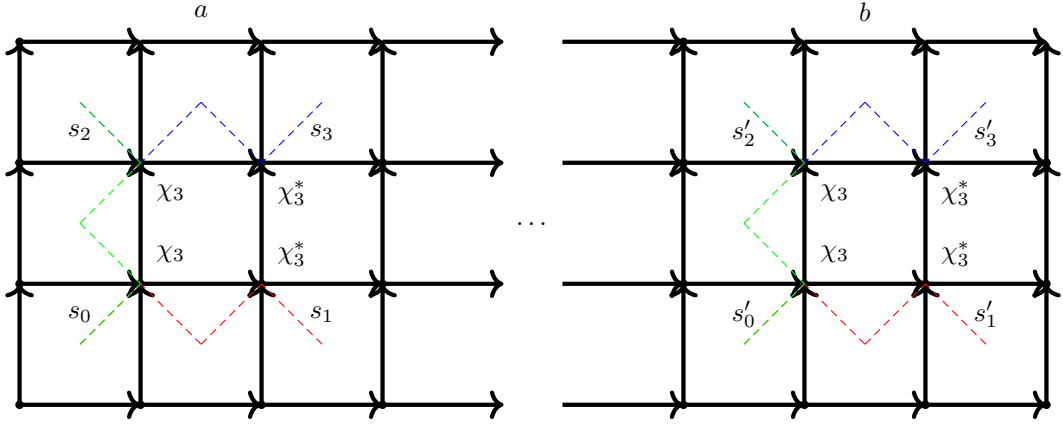
\[
\begin{tikzpicture}[scale=0.8]
    \draw (0,0) rectangle (6,6);
    \draw (0,2) -- (6,2);
        \draw (0,4) -- (6,4);
    \draw (2,0) -- (2,6);
        \draw (4,0) -- (4,6);
    \fill (0,0) circle (2pt);
    \fill (2,0) circle (2pt);
    \fill (4,0) circle (2pt);
    \fill (6,0) circle (2pt);
    \fill (0,2) circle (2pt);
    \fill (2,2) circle (2pt);
    \fill (4,2) circle (2pt);
    \fill (6,2) circle (2pt);
        \fill (0,4) circle (2pt);
                \fill (0,6) circle (2pt);
    \fill (2,4) circle (2pt);
    \fill (4,4) circle (2pt);
    \fill (6,4) circle (2pt);
        \fill (0,4) circle (2pt);
    \fill (2,4) circle (2pt);
    \fill (4,4) circle (2pt);
    \fill (6,4) circle (2pt);
    \draw [->,ultra thick](0,0) -- (2,0);
    \draw [->,ultra thick](2,0) -- (4,0);
    \draw [->,ultra thick](4,0) -- (6,0);
    \draw [->,ultra thick](0,2) -- (2,2);
    \draw [->,ultra thick](2,2) -- (4,2);
    \draw [->,ultra thick](4,2) -- (6,2);
    \draw [->,ultra thick](0,4) -- (2,4);
    \draw [->,ultra thick](2,4) -- (4,4);
    \draw [->,ultra thick](4,4) -- (6,4);
    \draw [->,ultra thick](0,6) -- (2,6);
    \draw [->,ultra thick](2,6) -- (4,6);
    \draw [->,ultra thick](4,6) -- (6,6);
    \draw [->,ultra thick](0,0) -- (0,2);
    \draw [->,ultra thick](0,2) -- (0,4);
    \draw [->,ultra thick](0,4) -- (0,6);
    \draw [->,ultra thick](2,0) -- (2,2);
    \draw [->,ultra thick](2,2) -- (2,4);
    \draw [->,ultra thick](2,4) -- (2,6);
    \draw [->,ultra thick](4,0) -- (4,2);
    \draw [->,ultra thick](4,2) -- (4,4);
    \draw [->,ultra thick](4,4) -- (4,6);
    \draw [->,ultra thick](6,0) -- (6,2);
    \draw [->,ultra thick](6,2) -- (6,4);
    \draw [->,ultra thick](6,4) -- (6,6);
    \draw [->,ultra thick](6,0) -- (8,0);
    \draw [->,ultra thick](6,2) -- (8,2);
    \draw [->,ultra thick](6,4) -- (8,4);
    \draw [->,ultra thick](6,6) -- (8,6);
    \draw[densely dashed,red] (1,1) -- (2,2);
    \draw[densely dashed,red] (3,1) -- (2,2);
    \draw[densely dashed,red] (3,1) -- (4,2);
    \draw[densely dashed,red] (5,1) -- (4,2);
     \node at (1, 1.5) {$s_0$};
     \node at (5, 1.5) {$s_1$};
     \node at (2.5, 2.5) {$\chi_3$};
     \node at (4.5, 2.5) {$\chi_3^*$};
    \draw[densely dashed,blue] (1,5) -- (2,4);
    \draw[densely dashed,blue] (3,5) -- (2,4);
    \draw[densely dashed,blue] (3,5) -- (4,4);
    \draw[densely dashed,blue] (5,5) -- (4,4);
     \node at (1, 4.5) {$s_2$};
     \node at (5, 4.5) {$s_3$};
     \node at (2.5, 3.5) {$\chi_3$};
     \node at (4.5, 3.5) {$\chi_3^*$};
    \draw[densely dashed,green] (1,3) -- (2,2);
    \draw[densely dashed,green] (1,3) -- (2,4);
        \draw[densely dashed,green] (1,1) -- (2,2);
            \draw[densely dashed,green] (1,5) -- (2,4);
    \node at (3,6.5) {$a$};
    \node at (8.5,3) [ultra thick]{$\dots$};

\end{tikzpicture}
\begin{tikzpicture}[scale=0.8][xshift=12]
    \draw (0,0) rectangle (6,6);
    \draw (0,2) -- (6,2);
        \draw (0,4) -- (6,4);
    \draw (2,0) -- (2,6);
        \draw (4,0) -- (4,6);
    \fill (0,0) circle (2pt);
    \fill (2,0) circle (2pt);
    \fill (4,0) circle (2pt);
    \fill (6,0) circle (2pt);
    \fill (0,2) circle (2pt);
    \fill (2,2) circle (2pt);
    \fill (4,2) circle (2pt);
    \fill (6,2) circle (2pt);
        \fill (0,4) circle (2pt);
                \fill (0,6) circle (2pt);
    \fill (2,4) circle (2pt);
    \fill (4,4) circle (2pt);
    \fill (6,4) circle (2pt);
        \fill (0,4) circle (2pt);
    \fill (2,4) circle (2pt);
    \fill (4,4) circle (2pt);
    \fill (6,4) circle (2pt);
    \draw [->,ultra thick](0,0) -- (2,0);
    \draw [->,ultra thick](2,0) -- (4,0);
    \draw [->,ultra thick](4,0) -- (6,0);
    \draw [->,ultra thick](0,2) -- (2,2);
    \draw [->,ultra thick](2,2) -- (4,2);
    \draw [->,ultra thick](4,2) -- (6,2);
    \draw [->,ultra thick](0,4) -- (2,4);
    \draw [->,ultra thick](2,4) -- (4,4);
    \draw [->,ultra thick](4,4) -- (6,4);
    \draw [->,ultra thick](0,6) -- (2,6);
    \draw [->,ultra thick](2,6) -- (4,6);
    \draw [->,ultra thick](4,6) -- (6,6);
    \draw [->,ultra thick](0,0) -- (0,2);
    \draw [->,ultra thick](0,2) -- (0,4);
    \draw [->,ultra thick](0,4) -- (0,6);
    \draw [->,ultra thick](2,0) -- (2,2);
    \draw [->,ultra thick](2,2) -- (2,4);
    \draw [->,ultra thick](2,4) -- (2,6);
    \draw [->,ultra thick](4,0) -- (4,2);
    \draw [->,ultra thick](4,2) -- (4,4);
    \draw [->,ultra thick](4,4) -- (4,6);
    \draw [->,ultra thick](6,0) -- (6,2);
    \draw [->,ultra thick](6,2) -- (6,4);
    \draw [->,ultra thick](6,4) -- (6,6);
    \draw [->,ultra thick](-2,0) -- (0,0);
    \draw [->,ultra thick](-2,2) -- (0,2);
    \draw [->,ultra thick](-2,4) -- (0,4);
    \draw [->,ultra thick](-2,6) -- (0,6);
    \draw[densely dashed,red] (1,1) -- (2,2);
    \draw[densely dashed,red] (3,1) -- (2,2);
    \draw[densely dashed,red] (3,1) -- (4,2);
    \draw[densely dashed,red] (5,1) -- (4,2);
     \node at (1, 1.5) {$s_0'$};
     \node at (5, 1.5) {$s_1'$};
     \node at (2.5, 2.5) {$\chi_3$};
     \node at (4.5, 2.5) {$\chi_3^*$};
    \draw[densely dashed,blue] (1,5) -- (2,4);
    \draw[densely dashed,blue] (3,5) -- (2,4);
    \draw[densely dashed,blue] (3,5) -- (4,4);
    \draw[densely dashed,blue] (5,5) -- (4,4);
     \node at (1, 4.5) {$s_2'$};
     \node at (5, 4.5) {$s_3'$};
     \node at (2.5, 3.5) {$\chi_3$};
     \node at (4.5, 3.5) {$\chi_3^*$};
    \draw[densely dashed,green] (1,3) -- (2,2);
    \draw[densely dashed,green] (1,3) -- (2,4);
        \draw[densely dashed,green] (1,1) -- (2,2);
            \draw[densely dashed,green] (1,5) -- (2,4);
    \node at (3,6.5) {$b$};

\end{tikzpicture}
\]
\begin{equation*}
\begin{aligned}
K_{a,b}=\frac{1}{3}\bigl(&I_a\otimes I_b+X_a\otimes I_b+X_a^2\otimes I_b\\
&+I_a\otimes X_b^2+\omega X_a\otimes X_b^2+\omega^2X_a^2\otimes X_b^2\\
&+I_a\otimes X_b+\omega^2X_a\otimes X_b+\omega X_a^2\otimes X_b\bigr).
\end{aligned}
\end{equation*}
where $X_a$ and $X_b$ are the logical operators on the respective qutrits. The operator $K_{a,b}$ has the following expressions as a quantum circuit and a ZX-diagram:

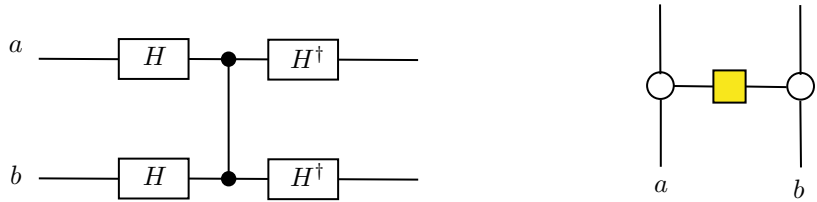
\[
\tikzset{every picture/.style={line width=0.75pt}} 
\begin{tikzpicture}[x=0.75pt,y=0.75pt,yscale=-1,xscale=1]
\draw    (219.96,149.61) -- (260.12,149.81) ;
\draw   (260.12,139.72) -- (295.12,139.72) -- (295.12,159.72) -- (260.12,159.72) -- cycle ;
\draw    (295.12,149.72) -- (335.28,149.91) ;
\draw   (335.12,140.05) -- (370.12,140.05) -- (370.12,160.05) -- (335.12,160.05) -- cycle ;
\draw    (370.12,150.05) -- (410.28,150.25) ;
\draw    (220,209.6) -- (260.16,209.8) ;
\draw   (260.16,199.7) -- (295.16,199.7) -- (295.16,219.7) -- (260.16,219.7) -- cycle ;
\draw    (295.16,209.7) -- (335.33,209.9) ;
\draw   (335.16,200.03) -- (370.16,200.03) -- (370.16,220.03) -- (335.16,220.03) -- cycle ;
\draw    (370.16,210.03) -- (410.33,210.23) ;
\draw    (315.2,149.81) -- (315.24,209.8) ;
\draw [shift={(315.24,209.8)}, rotate = 89.96] [color={rgb, 255:red, 0; green, 0; blue, 0 }  ][fill={rgb, 255:red, 0; green, 0; blue, 0 }  ][line width=0.75]      (0, 0) circle [x radius= 3.35, y radius= 3.35]   ;
\draw [shift={(315.2,149.81)}, rotate = 89.96] [color={rgb, 255:red, 0; green, 0; blue, 0 }  ][fill={rgb, 255:red, 0; green, 0; blue, 0 }  ][line width=0.75]      (0, 0) circle [x radius= 3.35, y radius= 3.35]   ;
\draw (270.63,142.25) node [anchor=north west][inner sep=0.75pt]    {$H$};
\draw (344.63,141.58) node [anchor=north west][inner sep=0.75pt]    {$H^{\dagger }$};
\draw (270.67,202.23) node [anchor=north west][inner sep=0.75pt]    {$H$};
\draw (203.53,138.81) node [anchor=north west][inner sep=0.75pt]    {$a$};
\draw (204.2,201.59) node [anchor=north west][inner sep=0.75pt]    {$b$};
\draw (344.63,200.58) node [anchor=north west][inner sep=0.75pt]    {$H^{\dagger }$};
\end{tikzpicture}
\hspace{3cm}
\begin{tikzpicture}[x=0.75pt,y=0.75pt,yscale=-1,xscale=1]
\draw    (210.28,216.93) -- (210.28,250.61) ;
\draw [color={rgb, 255:red, 0; green, 0; blue, 0 }  ,draw opacity=1 ][fill={rgb, 255:red, 0; green, 0; blue, 0 }  ,fill opacity=1 ]   (209.91,169.15) -- (210.07,204.18) ;
\draw [shift={(210.1,209.88)}, rotate = 89.74] [color={rgb, 255:red, 0; green, 0; blue, 0 }  ,draw opacity=1 ][line width=0.75]      (0, 0) circle [x radius= 6.7, y radius= 6.7]   ;
\draw    (280.28,217.5) -- (280.61,250.94) ;
\draw [color={rgb, 255:red, 0; green, 0; blue, 0 }  ,draw opacity=1 ][fill={rgb, 255:red, 0; green, 0; blue, 0 }  ,fill opacity=1 ]   (280.25,169.49) -- (280.4,204.51) ;
\draw [shift={(280.43,210.21)}, rotate = 89.74] [color={rgb, 255:red, 0; green, 0; blue, 0 }  ,draw opacity=1 ][line width=0.75]      (0, 0) circle [x radius= 6.7, y radius= 6.7]   ;
\draw    (216.56,210.08) -- (236.67,210.34) ;
\draw    (252.83,210.34) -- (273.42,210.65) ;
\draw  [fill={rgb, 255:red, 248; green, 231; blue, 28 }  ,fill opacity=1 ] (236.67,202.26) -- (252.83,202.26) -- (252.83,218.42) -- (236.67,218.42) -- cycle ;
\draw (205.53,255.81) node [anchor=north west][inner sep=0.75pt]    {$a$};
\draw (275.2,255.59) node [anchor=north west][inner sep=0.75pt]    {$b$};
\end{tikzpicture}
\]
In terms of ribbon operators, this is expressed as:
\[
\begin{aligned}
K_{a,b} = \frac{1}{3} \big( &I_a\otimes I_b + W_{\xi''_a}^{\chi_1}\otimes I_b + W_{\xi''_a}^{\chi_2}\otimes I_b + I_a\otimes W_{\xi''_b}^{\chi_2} + \omega W_{\xi''_a}^{\chi_1}\otimes W_{\xi''_b}^{\chi_2} \\
&+\omega^2 W_{\xi''_a}^{\chi_2}\otimes W_{\xi''_b}^{\chi_2} + I_a\otimes W_{\xi''_b}^{\chi_1} + \omega^2 W_{\xi''_a}^{\chi_1}\otimes W_{\xi''_b}^{\chi_1} + \omega W_{\xi''_a}^{\chi_2}\otimes W_{\xi''_b}^{\chi_1} \big).
\end{aligned}
\]

One can verify that $K_{a,b}$ functions as a controlled gate.

Unlike the qubit case, the logical operator $X_L$ here is not Hermitian. Therefore, to define continuous rotations, we must construct a Hermitian generator:
\[
H_X = i \frac{2\sqrt{3}\pi}{9} (X_L^2 - X_L).
\]

It is straightforward to verify that $H_X$ is indeed a Hermitian operator, which generates the continuous rotation:
\[
\begin{aligned}
R_X(\theta) &= \exp\left(-i \theta H_X\right) \\
&= \frac{1}{3} \left( 1 + e^{-i\theta \frac{2\pi}{3}} + e^{i\theta \frac{2\pi}{3}} \right) I
+ \frac{1}{3} \left( 1 + \omega^2 e^{-i\theta\frac{2\pi}{3}} + \omega e^{i\theta \frac{2\pi}{3}} \right) X_L \\
&\quad + \frac{1}{3} \left( 1 + \omega e^{-i\theta \frac{2\pi}{3}} + \omega^2 e^{i\theta \frac{2\pi}{3}} \right) X_L^2.
\end{aligned}
\]

In terms of ribbon operators, this rotation is given by:
\[
\begin{aligned}
R_X(\theta) &= \frac{1}{3} \left( 1 + e^{-i\theta \frac{2\pi}{3}} + e^{i\theta \frac{2\pi}{3}} \right) I
+ \frac{1}{3} \left( 1 + \omega^2 e^{-i\theta\frac{2\pi}{3}} + \omega e^{i\theta \frac{2\pi}{3}} \right) W_{\xi''}^{\chi_1} \\
&\quad + \frac{1}{3} \left( 1 + \omega e^{-i\theta \frac{2\pi}{3}} + \omega^2 e^{i\theta \frac{2\pi}{3}} \right) W_{\xi''}^{\chi_2}.
\end{aligned}
\]

Furthermore, we can realize the logical Hadamard gate. We first initialize the ancilla in the state $|0_b\rangle$.

If the measurement yields $\ket{0_a}$
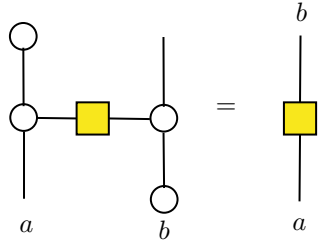
\[
\tikzset{every picture/.style={line width=0.75pt}} 
\begin{tikzpicture}[x=0.75pt,y=0.75pt,yscale=-1,xscale=1]
\draw    (191.28,155.19) -- (191.28,188.87) ;
\draw [color={rgb, 255:red, 0; green, 0; blue, 0 }  ,draw opacity=1 ][fill={rgb, 255:red, 0; green, 0; blue, 0 }  ,fill opacity=1 ]   (190.94,113.11) -- (191.07,142.44) ;
\draw [shift={(191.1,148.14)}, rotate = 89.74] [color={rgb, 255:red, 0; green, 0; blue, 0 }  ,draw opacity=1 ][line width=0.75]      (0, 0) circle [x radius= 6.7, y radius= 6.7]   ;
\draw [shift={(190.91,107.41)}, rotate = 89.74] [color={rgb, 255:red, 0; green, 0; blue, 0 }  ,draw opacity=1 ][line width=0.75]      (0, 0) circle [x radius= 6.7, y radius= 6.7]   ;
\draw    (261.28,155.76) -- (261.55,183.5) ;
\draw [shift={(261.61,189.2)}, rotate = 89.43] [color={rgb, 255:red, 0; green, 0; blue, 0 }  ][line width=0.75]      (0, 0) circle [x radius= 6.7, y radius= 6.7]   ;
\draw [color={rgb, 255:red, 0; green, 0; blue, 0 }  ,draw opacity=1 ][fill={rgb, 255:red, 0; green, 0; blue, 0 }  ,fill opacity=1 ]   (261.25,107.75) -- (261.4,142.77) ;
\draw [shift={(261.43,148.47)}, rotate = 89.74] [color={rgb, 255:red, 0; green, 0; blue, 0 }  ,draw opacity=1 ][line width=0.75]      (0, 0) circle [x radius= 6.7, y radius= 6.7]   ;
\draw    (197.56,148.34) -- (217.67,148.6) ;
\draw    (233.83,148.6) -- (254.42,148.91) ;
\draw  [fill={rgb, 255:red, 248; green, 231; blue, 28 }  ,fill opacity=1 ] (217.67,140.52) -- (233.83,140.52) -- (233.83,156.68) -- (217.67,156.68) -- cycle ;
\draw  [fill={rgb, 255:red, 248; green, 231; blue, 28 }  ,fill opacity=1 ] (321.67,140.52) -- (337.83,140.52) -- (337.83,156.68) -- (321.67,156.68) -- cycle ;
\draw    (330,107) -- (329.75,140.52) ;
\draw    (329.75,156.68) -- (329.5,190.2) ;
\draw (286,139.4) node [anchor=north west][inner sep=0.75pt]    {$=$};
\draw (187.53,198.4) node [anchor=north west][inner sep=0.75pt]    {$a$};
\draw (257.2,198.18) node [anchor=north west][inner sep=0.75pt]    {$b$};
\draw (324.53,197.4) node [anchor=north west][inner sep=0.75pt]    {$a$};
\draw (326.2,89.18) node [anchor=north west][inner sep=0.75pt]    {$b$};
\end{tikzpicture}
\]
Then, for an input state $|\psi\rangle$ on $a$, we obtain the output $H|\psi\rangle$ on $b$.

If the measurement yields $|1_a\rangle$, an input state $|\psi\rangle$ on $a$ produces the output $H R_X(1)|\psi\rangle$ on $b$.

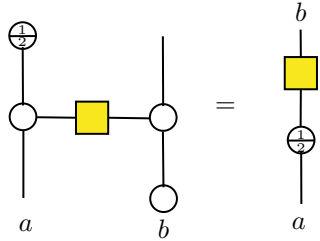
\[
\tikzset{every picture/.style={line width=0.75pt}} 
\begin{tikzpicture}[x=0.75pt,y=0.75pt,yscale=-1,xscale=1]
\draw    (87.28,107.19) -- (87.28,140.87) ;
\draw [color={rgb, 255:red, 0; green, 0; blue, 0 }  ,draw opacity=1 ][fill={rgb, 255:red, 0; green, 0; blue, 0 }  ,fill opacity=1 ]   (86.94,65.11) -- (87.07,94.44) ;
\draw [shift={(87.1,100.14)}, rotate = 89.74] [color={rgb, 255:red, 0; green, 0; blue, 0 }  ,draw opacity=1 ][line width=0.75]      (0, 0) circle [x radius= 6.7, y radius= 6.7]   ;
\draw [shift={(86.91,59.41)}, rotate = 89.74] [color={rgb, 255:red, 0; green, 0; blue, 0 }  ,draw opacity=1 ][line width=0.75]      (0, 0) circle [x radius= 6.7, y radius= 6.7]   ;
\draw    (157.28,107.76) -- (157.55,135.5) ;
\draw [shift={(157.61,141.2)}, rotate = 89.43] [color={rgb, 255:red, 0; green, 0; blue, 0 }  ][line width=0.75]      (0, 0) circle [x radius= 6.7, y radius= 6.7]   ;
\draw [color={rgb, 255:red, 0; green, 0; blue, 0 }  ,draw opacity=1 ][fill={rgb, 255:red, 0; green, 0; blue, 0 }  ,fill opacity=1 ]   (157.25,59.75) -- (157.4,94.77) ;
\draw [shift={(157.43,100.47)}, rotate = 89.74] [color={rgb, 255:red, 0; green, 0; blue, 0 }  ,draw opacity=1 ][line width=0.75]      (0, 0) circle [x radius= 6.7, y radius= 6.7]   ;
\draw    (93.56,100.34) -- (113.67,100.6) ;
\draw    (129.83,100.6) -- (150.42,100.91) ;
\draw  [fill={rgb, 255:red, 248; green, 231; blue, 28 }  ,fill opacity=1 ] (113.67,92.52) -- (129.83,92.52) -- (129.83,108.68) -- (113.67,108.68) -- cycle ;
\draw  [fill={rgb, 255:red, 248; green, 231; blue, 28 }  ,fill opacity=1 ] (218.39,70.16) -- (234.56,70.16) -- (234.56,86.32) -- (218.39,86.32) -- cycle ;
\draw    (226.36,56.4) -- (226.47,70.16) ;
\draw    (226.47,86.32) -- (226.36,105.31) ;
\draw    (80.08,59.68) -- (93.7,59.62) ;
\draw [color={rgb, 255:red, 0; green, 0; blue, 0 }  ,draw opacity=1 ][fill={rgb, 255:red, 0; green, 0; blue, 0 }  ,fill opacity=1 ]   (226.64,117.5) -- (226.73,143.86) ;
\draw [shift={(226.62,111.8)}, rotate = 89.81] [color={rgb, 255:red, 0; green, 0; blue, 0 }  ,draw opacity=1 ][line width=0.75]      (0, 0) circle [x radius= 6.7, y radius= 6.7]   ;
\draw    (219.79,112.07) -- (233.41,112) ;
\draw (182,91.4) node [anchor=north west][inner sep=0.75pt]    {$=$};
\draw (83.53,150.4) node [anchor=north west][inner sep=0.75pt]    {$a$};
\draw (153.2,150.18) node [anchor=north west][inner sep=0.75pt]    {$b$};
\draw (220.53,149.4) node [anchor=north west][inner sep=0.75pt]    {$a$};
\draw (222.2,41.18) node [anchor=north west][inner sep=0.75pt]    {$b$};
\draw (83,53) node [anchor=north west][inner sep=0.75pt]  [font=\fontsize{0.35em}{0.42em}\selectfont]  {$1$};
\draw (83,59) node [anchor=north west][inner sep=0.75pt]  [font=\fontsize{0.35em}{0.42em}\selectfont]  {$2$};
\draw (222.78,105.32) node [anchor=north west][inner sep=0.75pt]  [font=\fontsize{0.35em}{0.42em}\selectfont]  {$1$};
\draw (222.72,112.21) node [anchor=north west][inner sep=0.75pt]  [font=\fontsize{0.35em}{0.42em}\selectfont]  {$2$};
\end{tikzpicture}
\]

If the measurement yields $|2_a\rangle$, an input state $|\psi\rangle$ on $a$ produces the output $H R_X(2)|\psi\rangle$ on $b$.

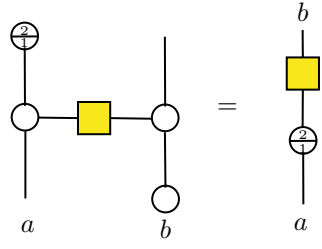
\[
\tikzset{every picture/.style={line width=0.75pt}} 
\begin{tikzpicture}[x=0.75pt,y=0.75pt,yscale=-1,xscale=1]
\draw    (107.28,127.19) -- (107.28,160.87) ;
\draw [color={rgb, 255:red, 0; green, 0; blue, 0 }  ,draw opacity=1 ][fill={rgb, 255:red, 0; green, 0; blue, 0 }  ,fill opacity=1 ]   (106.94,85.11) -- (107.07,114.44) ;
\draw [shift={(107.1,120.14)}, rotate = 89.74] [color={rgb, 255:red, 0; green, 0; blue, 0 }  ,draw opacity=1 ][line width=0.75]      (0, 0) circle [x radius= 6.7, y radius= 6.7]   ;
\draw [shift={(106.91,79.41)}, rotate = 89.74] [color={rgb, 255:red, 0; green, 0; blue, 0 }  ,draw opacity=1 ][line width=0.75]      (0, 0) circle [x radius= 6.7, y radius= 6.7]   ;
\draw    (177.28,127.76) -- (177.55,155.5) ;
\draw [shift={(177.61,161.2)}, rotate = 89.43] [color={rgb, 255:red, 0; green, 0; blue, 0 }  ][line width=0.75]      (0, 0) circle [x radius= 6.7, y radius= 6.7]   ;
\draw [color={rgb, 255:red, 0; green, 0; blue, 0 }  ,draw opacity=1 ][fill={rgb, 255:red, 0; green, 0; blue, 0 }  ,fill opacity=1 ]   (177.25,79.75) -- (177.4,114.77) ;
\draw [shift={(177.43,120.47)}, rotate = 89.74] [color={rgb, 255:red, 0; green, 0; blue, 0 }  ,draw opacity=1 ][line width=0.75]      (0, 0) circle [x radius= 6.7, y radius= 6.7]   ;
\draw    (113.56,120.34) -- (133.67,120.6) ;
\draw    (149.83,120.6) -- (170.42,120.91) ;
\draw  [fill={rgb, 255:red, 248; green, 231; blue, 28 }  ,fill opacity=1 ] (133.67,112.52) -- (149.83,112.52) -- (149.83,128.68) -- (133.67,128.68) -- cycle ;
\draw  [fill={rgb, 255:red, 248; green, 231; blue, 28 }  ,fill opacity=1 ] (238.39,90.16) -- (254.56,90.16) -- (254.56,106.32) -- (238.39,106.32) -- cycle ;
\draw    (246.36,76.4) -- (246.47,90.16) ;
\draw    (246.47,106.32) -- (246.36,125.31) ;
\draw    (100.08,79.68) -- (113.7,79.62) ;
\draw [color={rgb, 255:red, 0; green, 0; blue, 0 }  ,draw opacity=1 ][fill={rgb, 255:red, 0; green, 0; blue, 0 }  ,fill opacity=1 ]   (246.64,137.5) -- (246.73,163.86) ;
\draw [shift={(246.62,131.8)}, rotate = 89.81] [color={rgb, 255:red, 0; green, 0; blue, 0 }  ,draw opacity=1 ][line width=0.75]      (0, 0) circle [x radius= 6.7, y radius= 6.7]   ;
\draw    (239.79,132.07) -- (253.41,132) ;
\draw (202,111.4) node [anchor=north west][inner sep=0.75pt]    {$=$};
\draw (103.53,170.4) node [anchor=north west][inner sep=0.75pt]    {$a$};
\draw (173.2,170.18) node [anchor=north west][inner sep=0.75pt]    {$b$};
\draw (240.53,169.4) node [anchor=north west][inner sep=0.75pt]    {$a$};
\draw (242.2,61.18) node [anchor=north west][inner sep=0.75pt]    {$b$};
\draw (103,73) node [anchor=north west][inner sep=0.75pt]  [font=\fontsize{0.35em}{0.42em}\selectfont]  {$2$};
\draw (103,79) node [anchor=north west][inner sep=0.75pt]  [font=\fontsize{0.35em}{0.42em}\selectfont]  {$1$};
\draw (242.78,125.32) node [anchor=north west][inner sep=0.75pt]  [font=\fontsize{0.35em}{0.42em}\selectfont]  {$2$};
\draw (242.72,132.21) node [anchor=north west][inner sep=0.75pt]  [font=\fontsize{0.35em}{0.42em}\selectfont]  {$1$};
\end{tikzpicture}
\]
We can repeat the above measurement protocol until the desired Hadamard gate is implemented, with any necessary corrections achieved via the rotation gates.

Equipped with the logical Hadamard gate and $X$ rotations, we can act transitively on the qutrit pure-state space $\mathbb{CP}^2$. Furthermore, the addition of the entangling gate $K_{a,b}$ allows for the implementation of arbitrary unitary operations. Relying on the operators provided by the $D(A_4)$ model, we have thus realized a scheme for universal quantum computation.

\begin{appendix}

\section{Representation data for \texorpdfstring{$A_4$}{A4} and \texorpdfstring{$D(A_4)$}{D(A4)}}\label{app:a4}
\subsection{Representations of the alternating group \texorpdfstring{$A_4$}{A4}}

\( A_4 \) has 12 elements, classified as follows:

Identity element: \( e \)

 Order-2 elements (Klein four-group):
  \[
  u_1 = (12)(34), \quad u_2 = (13)(24), \quad u_3 = (14)(23).
  \]

 Order-3 elements:
  \[
  \begin{aligned}
  & v_1 = (123), \quad v_2 = (132), \quad v_3 = (124), \quad v_4 = (142), \\
  & v_5 = (134), \quad v_6 = (143), \quad v_7 = (234), \quad v_8 = (243).
  \end{aligned}
  \]

The group multiplication table is given below (rows act on columns):
\[
\small
\begin{array}{|c|cccccccccccc|}
\hline
 & e & u_1 & u_2 & u_3 & v_1 & v_2 & v_3 & v_4 & v_5 & v_6 & v_7 & v_8 \\
\hline
e & e & u_1 & u_2 & u_3 & v_1 & v_2 & v_3 & v_4 & v_5 & v_6 & v_7 & v_8 \\
u_1 & u_1 & e & u_3 & u_2 & v_8 & v_6 & v_7 & v_5 & v_4 & v_2 & v_3 & v_1 \\
u_2 & u_2 & u_3 & e & u_1 & v_4 & v_7 & v_6 & v_1 & v_8 & v_3 & v_2 & v_5 \\
u_3 & u_3 & u_2 & u_1 & e & v_5 & v_3 & v_2 & v_8 & v_1 & v_7 & v_6 & v_4 \\
v_1 & v_1 & v_5 & v_8 & v_4 & v_2 & e & u_2 & v_6 & v_7 & u_3 & u_1 & v_3 \\
v_2 & v_2 & v_7 & v_3 & v_6 & e & v_1 & v_8 & u_3 & u_1 & v_4 & v_5 & u_2 \\
v_3 & v_3 & v_6 & v_2 & v_7 & u_3 & v_5 & v_4 & e & u_2 & v_8 & v_1 & u_1 \\
v_4 & v_4 & v_8 & v_5 & v_1 & v_7 & u_2 & e & v_3 & v_2 & u_1 & u_3 & v_6 \\
v_5 & v_5 & v_1 & v_4 & v_8 & v_3 & u_3 & u_1 & v_7 & v_6 & e & u_2 & v_2 \\
v_6 & v_6 & v_3 & v_7 & v_2 & u_1 & v_8 & v_1 & u_2 & e & v_5 & v_4 & u_3 \\
v_7 & v_7 & v_2 & v_6 & v_3 & u_2 & v_4 & v_5 & u_1 & u_3 & v_1 & v_8 & e \\
v_8 & v_8 & v_4 & v_1 & v_5 & v_6 & u_1 & u_3 & v_2 & v_3 & u_2 & e & v_7 \\
\hline
\end{array}
\]

We observe the conjugation relations:
\[
\begin{aligned}
&v_1u_1v_1^{-1} = u_3, \quad &v_1u_3v_1^{-1} = u_2, \quad &v_1u_2v_1^{-1} = u_1,\\
& v_1 v_3 v_1^{-1} = v_7, \quad
&v_1 v_7 v_1^{-1} = v_6 ,\quad
&v_1 v_6 v_1^{-1} = v_3 ,\\
& v_1 v_4 v_1^{-1} = v_8 ,\quad
&v_1 v_8 v_1^{-1} = v_5 ,\quad
&v_1 v_5 v_1^{-1} = v_4 .
\end{aligned}
\]
This reveals that \( A_4 \) has four conjugacy classes:
\[
\begin{aligned}
&\{1\}, \qquad \{u_1,u_2,u_3\},\\
&\{v_1,\,v_1u_1(=v_5),\,v_1u_2(=v_8),\,v_1u_3(=v_4)\},\\
&\{v_1^2(=v_2),\,v_1^2u_1(=v_7),\,v_1^2u_2(=v_3),\,v_1^2u_3(=v_6)\}.
\end{aligned}
\]
Thus, \( A_4 \) has four irreducible representations \( \chi_0, \chi_1, \chi_2, \chi_3 \).

The commutator subgroup of \( A_4 \) is computed as:
\[
\begin{aligned}
\bigl[A_4,A_4\bigr] &= \langle ghg^{-1}h^{-1} \mid g,h \in A_4\rangle\\
&= \{e, u_1, u_2, u_3\} := U,
\end{aligned}
\]
and \( A_4/[A_4,A_4] \cong \mathbb{Z}_3 \). Consequently, \( A_4 \) has three one-dimensional representations inherited from \( \mathbb{Z}_3 \), where the characters are trivial on \( U \):
\[
\begin{array}{|c|cccc|}
\hline
 & e & u_1 & v_1 & v_1^2 \\
\hline
\chi_0 & 1 & 1 & 1 & 1 \\
\chi_1 & 1 & 1 & \omega & \omega^2 \\
\chi_2 & 1 & 1 & \omega^2 & \omega \\
\hline
\end{array}
\]
where \( \omega = e^{2\pi i / 3} \).

The fourth irreducible representation \( \chi_3 \) is derived via orthogonality relations:
\[
\sum_{i=0}^{3} \chi_i(1) \overline{\chi_i(1)} = 12, \quad \sum_{i=0}^{3} \chi_i(u_1) \overline{\chi_i(u_1)} = 4, \quad \sum_{i=0}^{3} \chi_i(v_1) \overline{\chi_i(v_1)} = \sum_{i=0}^{3} \chi_i(v_1^2) \overline{\chi_i(v_1^2)} = 3,
\]
and
\[
\sum_{i=0}^{3} \chi_i(1) \overline{\chi_i(u_1)} = \cdots = 0.
\]

The complete character table of \( A_4 \) is:
\[
\begin{array}{|c|cccc|}
\hline
 & e & u_1 & v_1 & v_1^2 \\
\hline
\chi_0 & 1 & 1 & 1 & 1 \\
\chi_1 & 1 & 1 & \omega & \omega^2 \\
\chi_2 & 1 & 1 & \omega^2 & \omega \\
\chi_3 & 3 & -1 & 0 & 0 \\
\hline
\end{array}
\]

The alternating group \( A_4 \) has irreducible representations \( \{\chi_0, \chi_1, \chi_2, \chi_3\} \).

The tensor product decompositions of these representations are:
\[
\begin{aligned}
\chi_1 \otimes \chi_1 &= \chi_2, \quad \chi_2 \otimes \chi_2 = \chi_1, \\
\chi_1 \otimes \chi_2 &= \chi_0, \quad \chi_1 \otimes \chi_3 = \chi_2 \otimes \chi_3 = \chi_3, \\
\chi_3 \otimes \chi_3 &= \chi_0 \oplus \chi_1 \oplus \chi_2 \oplus 2\chi_3.
\end{aligned}
\]

\subsection{Representations of \texorpdfstring{$D(A_4)$}{D(A4)}}

As established in Section~\ref{sec:prereq}, the irreducible representations of \( D(G) \) are determined by pairs \( (\mathcal{C}, \pi) \). We now compute the irreducible representations of \( D(A_4) \).

(1) Conjugacy Class \( \mathcal{C} = \{e\} \)
. Choose \( r_{\mathcal{C}} = q_e = e \), and \( C_G = A_4 \).
 This yields 4 irreducible representations
  \[
  (\{e\}, \chi_0), \quad (\{e\}, \chi_1), \quad (\{e\}, \chi_2), \quad (\{e\}, \chi_3).
  \]

(2) Conjugacy Class \( \mathcal{C} = \{u_1, u_2, u_3\} \) .
 Choose \( r_{\mathcal{C}} = u_1 \), \( q_{u_1} = e \).
From the relations:
  \[
  q_{u_2} u_1 q_{u_2}^{-1} = u_2, \quad q_{u_3} u_1 q_{u_3}^{-1} = u_3,
  \]
  we compute \( q_{u_2} = v_1^2 \), \( q_{u_3} = v_1 \). Here, \( C_G = \{e, u_1, u_2, u_3\} \cong \mathbb{Z}_2 \times \mathbb{Z}_2 \).

 Irreducible representations of \( \mathbb{Z}_2 \times \mathbb{Z}_2 \) :
 \[
  \begin{array}{|c|cccc|}
  \hline
  & e & u_1 & u_2 & u_3 \\
  \hline
  \chi_{00} & 1 & 1 & 1 & 1 \\
  \chi_{10} & 1 & -1 & 1 & -1 \\
  \chi_{01} & 1 & 1 & -1 & -1 \\
  \chi_{11} & 1 & -1 & -1 & 1 \\
  \hline
  \end{array}
  \]
 This gives 4 irreducible representations
  \[
  (\{u_1, u_2, u_3\}, \chi_{00}), \quad (\{u_1, u_2, u_3\}, \chi_{10}), \quad (\{u_1, u_2, u_3\}, \chi_{01}), \quad (\{u_1, u_2, u_3\}, \chi_{11}).
  \]

(3) Conjugacy Class \( \mathcal{C} = \{v_1, v_1 u_1, v_1 u_2, v_1 u_3\} \)
. Choose \( r_{\mathcal{C}} = v_1 \), \( q_{v_1} = e \).
From the relations:
  \[
  q_{v_1 u_1} v_1 q_{v_1 u_1}^{-1} = v_1 u_1, \quad q_{v_1 u_2} v_1 q_{v_1 u_2}^{-1} = v_1 u_2, \quad q_{v_1 u_3} v_1 q_{v_1 u_3}^{-1} = v_1 u_3,
  \]
  we compute \( q_{v_1 u_1} = u_2 \), \( q_{v_1 u_2} = u_3 \), \( q_{v_1 u_3} = u_1 \). Here, \( C_G = \{e, v_1, v_2\} \cong \mathbb{Z}_3 \).

 Irreducible representations of \( \mathbb{Z}_3 \) ( \( w = e^{2\pi i /3} \)):
  \[
  \begin{array}{|c|ccc|}
  \hline
  & e & v_1 & v_2 \\
  \hline
  1 & 1 & 1 & 1 \\
  \pi_{\omega_1} & 1 & \omega & \omega^2 \\
  \pi_{\omega_2} & 1 & \omega^2 & \omega \\
  \hline
  \end{array}
  \]
 This gives 3 irreducible representations
  \[
  (\{v_1, v_1 u_1, v_1 u_2, v_1 u_3\}, 1), \quad (\{v_1, v_1 u_1, v_1 u_2, v_1 u_3\}, \pi_{w_1}), \quad (\{v_1, v_1 u_1, v_1 u_2, v_1 u_3\}, \pi_{w_2}).
  \]

(4) Conjugacy Class \( \mathcal{C} = \{v_1^2, v_1^2 u_1, v_1^2 u_2, v_1^2 u_3\} \)
.Choose \( r_{\mathcal{C}} = v_1^2 \), \( q_{v_1^2} = e \).
From the relations:
  \[
  q_{v_1^2 u_1} v_1^2 q_{v_1^2 u_1}^{-1} = v_1^2 u_1, \quad q_{v_1^2 u_2} v_1^2 q_{v_1^2 u_2}^{-1} = v_1^2 u_2, \quad q_{v_1^2 u_3} v_1^2 q_{v_1^2 u_3}^{-1} = v_1^2 u_3,
  \]
  we compute \( q_{v_1^2 u_1} = u_3 \), \( q_{v_1^2 u_2} = u_1 \), \( q_{v_1^2 u_3} = u_2 \). Here, \( C_G = \{e, v_1, v_2\} \cong \mathbb{Z}_3 \), with the same representations as in (3).

 This gives 3 irreducible representations
  \[
  (\{v_1^2, v_1^2 u_1, v_1^2 u_2, v_1^2 u_3\}, 1), \quad (\{v_1^2, v_1^2 u_1, v_1^2 u_2, v_1^2 u_3\}, \pi_{w_1}), \quad (\{v_1^2, v_1^2 u_1, v_1^2 u_2, v_1^2 u_3\}, \pi_{w_2}).
  \]

The 14 irreducible representations of \( D(A_4) \) are
\[
\begin{aligned}
&\Big\{ (\{e\}, \chi_0), \ (\{e\}, \chi_1), \ (\{e\}, \chi_2), \ (\{e\}, \chi_3), \\
&(\{u_1, u_2, u_3\}, \chi_{00}), \ (\{u_1, u_2, u_3\}, \chi_{10}), \ (\{u_1, u_2, u_3\}, \chi_{01}), \ (\{u_1, u_2, u_3\}, \chi_{11}), \\
&(\{v_1, v_1 u_1, v_1 u_2, v_1 u_3\}, 1), \ (\{v_1, v_1 u_1, v_1 u_2, v_1 u_3\}, \pi_{w_1}), \ (\{v_1, v_1 u_1, v_1 u_2, v_1 u_3\}, \pi_{w_2}), \\
&(\{v_1^2, v_1^2 u_1, v_1^2 u_2, v_1^2 u_3\}, 1), \ (\{v_1^2, v_1^2 u_1, v_1^2 u_2, v_1^2 u_3\}, \pi_{w_1}), \ (\{v_1^2, v_1^2 u_1, v_1^2 u_2, v_1^2 u_3\}, \pi_{w_2}) \Big\}.
\end{aligned}
\]

\end{appendix}

\end{document}